\documentclass[a4paper,11pt]{article}
\usepackage[english]{babel}
\usepackage[a4paper]{geometry}
\usepackage{bbm}

\usepackage{comment}
 
\usepackage{mathrsfs}
\usepackage{enumerate}
 
\usepackage{hyperref}
\usepackage{amsthm,amsmath,amsfonts,amssymb,amsxtra,bookmark,dsfont,bm,xcolor,enumerate,mathtools}

\theoremstyle{plain}

\newtheorem{theorem}{Theorem}[section]

\newtheorem{proposition}[theorem]{Proposition}

\newtheorem{lemma}[theorem]{Lemma}

\newtheorem*{remark*}{Remark}

\theoremstyle{remark}

\def\tr{\text{tr}\,} 
\def\trf{\text{tr}_{\cF_+}}

\def\eps{\varepsilon}

\def\1{{\ensuremath {\mathds 1} }}

\def\bC{\mathbb{C}}

\def\bN{\mathbb{N}}

\def\bR{\mathbb{R}}

\def\bT{\mathbb{T}}

\def\bZ{\mathbb{Z}}

\def\cD{\mathcal{D}}
\def\cE{\mathcal{E}}
\def\cF{\mathcal{F}}
\def\cG{\mathcal{G}}
\def\cH{\mathcal{H}}

\def\cK{\mathcal{K}}
\def\cL{\mathcal{L}}
\def\cM{\mathcal{M}}
\def\cN{\mathcal{N}}

\def\cP{\mathcal{P}}
\def\cQ{\mathcal{Q}}

\def\cU{\mathcal{U}}
\def\cV{\mathcal{V}}

\def\tbf{\textbf }

\def\ii{\mathrm{i}}

\def\be{\begin{equation}}
\def\ee{\end{equation}}
\def\eps{\varepsilon}

\newcommand{\ph}{\varphi}

\newcommand{\wh}{\widehat}

\title{Second Order Expansion of Gibbs State Reduced {Density Matrices} in the Gross-Pitaevskii Regime}  

\author{Christian Brennecke\thanks{Institute for Applied Mathematics, University of Bonn, Endenicher Allee 60, 53115 Bonn, Germany.}\footnotemark[1] 
\and Jinyeop Lee\thanks{Department of Mathematics, LMU Munich, Theresienstrasse 39, 80333 Munich, Germany. \hspace{20 mm} Emails: brennecke@iam.uni-bonn.de, lee@math.lmu.de, nam@math.lmu.de}
\and Phan Th\`anh Nam\footnotemark[2] 
}

\begin{document}

\maketitle

\begin{abstract}
We consider a translation-invariant system of $N$ bosons in $\bT^{3}$ that interact through a repulsive two-body potential with scattering length of order $N^{-1}$ in the limit $N\to \infty$. We derive second order expressions for the one- and two-particle reduced density matrix matrices of the Gibbs state at fixed positive temperatures, thus obtaining a justification of Bogoliubov's prediction on the fluctuations around the condensate. 
%In recent years, rigorous versions \cite{BBCS3, NT, BSS2, HST} of Bogoliubov's method \cite{Bog} have been used to determine the low energy spectrum for such systems (and, more generally, for systems of bosons trapped by an external potential), up to errors that vanish in the limit $N\to \infty$. As a consequence, one can verify e.g. Bogoliubov's prediction on the condensate depletion in the ground state. In the translation invariant setting, this provides a second order approximation of its reduced one particle density matrix in the trace class topology and in this note, we derive analogous second order expressions for the one- and two-particle reduced density matricies of the Gibbs state at positive temperature. 
\end{abstract}

\section{Introduction}

We consider a system of $N$ bosons in the torus $\Lambda = \bT^{3} = \bR^3 / \bZ^3$ in the Gross-Pitaevskii (GP) regime. The system is described by Hamiltonian
\begin{equation}\label{eq:defHN} 
H_N = \sum_{i=1}^N (-\Delta_{x_i}) +  \sum_{1\leq i<j\leq N} N^2 V (N (x_i -x_j)), 
\end{equation}
on the bosonic space $L^2_s (\Lambda^N) = \bigotimes_{\text{sym}}^N L^2(\Lambda)$, a subspace of $L^2 (\Lambda^N)$ consisting of functions that are symmetric w.r.t. permutations for each pair of variables. The interaction potential $V\in L^3 (\bR^3)$ is assumed to be pointwise non-negative, radial and compactly supported. Under this condition, $H_N$ is a self-adjoint operator with the 
%same domain $\bigotimes_{\text{sym}}^N H^2(\Lambda)$ of the non-interacting one. 
same domain $\bigotimes_{\text{sym}}^N H^2(\Lambda)$ as that of the non-interacting Hamiltonian, by Kato's Theorem (see e.g.\ \cite{RS2}). 

The scaling $ N^2V(N \, \cdot\,)$ of the interaction in \eqref{eq:defHN} implies that its scattering length is of order $N^{-1}$, namely
$$ \mathfrak{a}_N \equiv \mathfrak{a}(N^2V(N\,\cdot\,))=N^{-1}\mathfrak{a} $$
where 
\begin{equation}\label{eq:a-def} \mathfrak{a}=\mathfrak{a}(V) = \frac1{8\pi} \inf \Big\{ \int_{\bR^3} 2|\nabla f|^2+V|f|^2: \lim_{|x|\to\infty}f(x)=1 \Big\}.
\end{equation}
Under our assumptions on $V$, the infimum in \eqref{eq:a-def} is attained for a unique minimizer $0\leq f\leq 1$ that solves the zero energy scattering equation 
        \begin{align} \label{eq:scattering-full-space} (-2\Delta + V)f = 0. \end{align}
%By scaling, this implies that $ \mathfrak{a}_N = \mathfrak{a}/N$, setting from now on $\mathfrak{a}\equiv \mathfrak{a}(V)$.

In the Gross-Pitaevskii regime, the scattering length of the interaction $\mathfrak{a}_N \sim N^{-1}$ is much smaller than the mean distance $O(N^{-1/3})$ between particles, but the interaction strength $\|N^2V(N \, \cdot\,)\|_{L^\infty} \sim N^2$  is sufficiently large to have a leading order effect to the spectral property of the system when $N\to \infty$. In particular, since the system is very dilute, any particle is essentially felt by the others as a hard sphere of radius $\mathfrak{a}_N$, and hence it is reasonable to apply the mean-field approximation with the formal replacement $N^2V(N \, \cdot\,)\sim 8\pi  \mathfrak{a}_N \delta (\cdot)$. However, justifying this  heuristic idea is difficult since subtle mathematical arguments have to be developed to understand strong correlations at short distances. This is among the reasons why the Gross-Pitaevskii regime has attracted significant research attention for more than three decades. We refer to the books \cite{LSSY,BPS} and the survey \cite{Rou} for  pedagogical introductions to this scaling limit. 

To the leading order, the ground state energy of $H_N$ was derived in \cite{LY,LSY} and the Bose--Einstein condensation (BEC) was proved in \cite{LS1,LS2}. In the later context, the one-body density matrix of the ground state $\Psi_N$ has the largest eigenvalue $1+ o(1)_{N\to \infty}$ with the corresponding approximate eigenfunction given by the zero momentum condensate wave function $\varphi_0\equiv 1\in L^2(\Lambda)$, namely
\begin{equation} \label{eq:leading-1pdm}
\gamma_N^{(1)}= {\rm tr}_{2\to N} |\Psi_N\rangle \langle \Psi_N| \to |\varphi_0\rangle \langle \varphi_0|
\end{equation}
in trace class. We also refer to \cite{NRS,ABS,H,NRT,BocSei} for alternative approaches to \eqref{eq:leading-1pdm}, \cite{BBCS1, BBCS4,NNRT,BSS1} for optimal error estimates, \cite{DSY, DS,NR,BDS} for related results at positive temperatures, and \cite{ESY1, ESY2, ESY3,P,BDS,BS,COS} for the dynamical analogue.

Recently, the second order of the low-lying eigenvalues of $H_N$ was derived in \cite{BBCS3}, which justified Bogoliubov's approximation \cite{Bog} for the excitation spectrum. In this direction, we also refer to \cite{HST} for a simplified proof,  \cite{BCaS,NT,BSS2,BCOPS,Brooks} for extensions to other trapped systems and beyond the Gross-Pitaevskii regime, and \cite{YY,FS1,FS2,BCS,HHNST} for related results in the thermodynamic limit where the Lee--Huang--Yang formula was rigorously investigated.

Note that the spectral analysis in \cite{BBCS3} also gives information on low-lying eigenstates of $H_N$. To be precise, as a consequence of the proof in \cite{BBCS3}, one obtains a norm approximation of the form (see \cite[Eq. (6.7)]{BBCS3})
        \be\label{eq:gsappr} \| \psi_N - \cU \varphi_0^{\otimes N} \|^2\leq C N^{-1/4}, \ee
where $\psi_N\in L^2_s(\Lambda^N)$ denotes the (unique and positive) ground state of $H_N$ and $ \cU$ denotes an explicit unitary map on $L^2_s(\Lambda^N)$ that takes into account relevant pair correlations among the particles (for more details, see Section \ref{sec:recap}, where we recall basic definitions and results from \cite{BBCS3} that are relevant for our purposes). Using \eqref{eq:gsappr}, it was proved moreover that 
        \be\label{eq:condepl} N - N \langle \ph_0, \gamma_N^{(1)}\ph_0\rangle = \sum_{p\in \Lambda^*_+} \mu_p^2 + O\big(N^{-1/8}\big),\ee
where we set from now on $\Lambda^*_+ =(2\pi \bZ)^3\setminus\{0\}$ and where
 \[\mu_p^2 = \frac{|p|^2+8\pi \mathfrak{a} - \sqrt{|p|^4+16\pi \mathfrak{a}|p|^2}}{2\sqrt{|p|^4+16\pi \mathfrak{a}|p|^2}}\geq 0 .\]
The formula \eqref{eq:condepl} on the condensate depletion is basic prediction that follows from Bogoliubov's method \cite{Bog}. In other words, defining the orthogonal projections $P_0 =|\ph_0\rangle\langle\ph_0| $ and  $Q_0 =1-P_0$, we have 
 \begin{align} \label{eq:CV-Qgamma1Q-trace}
 \lim_{N\to\infty} \tr \, \Big| N Q_0 \gamma_N^{(1)} Q_0-  \sum_{p\in \Lambda_+^*} \mu_p^2 |\ph_p\rangle\langle\ph_p|\Big|=0
 \end{align}
 where  $x\mapsto \ph_p(x) = e^{\ii px}\in L^2(\Lambda)$ denotes the plane wave of momentum $p$. The trace convergence \eqref{eq:CV-Qgamma1Q-trace} follows from the convergence of the traces \eqref{eq:condepl}, the weak-$*$ convergence of $N Q_0 \gamma_N^{(1)} Q_0$ (which can be proved like \eqref{eq:condepl}), the positivity $N Q_0 \gamma_N^{(1)} Q_0 \geq 0 $ and a general abstract argument (see e.g. \cite[Theorem 2.20]{Simon}). Together with $\tr\gamma_N^{(1)}=1$ and $ P_0 \gamma_N^{(1)} Q_0 = Q_0 \gamma_N^{(1)} P_0=0$ by translation invariance of $ \gamma_N^{(1)}$, one thus obtains the following refinement of \eqref{eq:leading-1pdm}:
 \begin{align} \label{eq:CV-1pdm-GS}  
 N\gamma_N^{(1)} = \bigg(N -\sum_{p\in\Lambda_+^*} \mu_p^2\bigg) |\ph_0\rangle\langle\ph_0| +\sum_{p\in\Lambda_+^*}\mu_p^2|\ph_p\rangle\langle\ph_p|  + o(1)_{N\to \infty} 
  \end{align}
where the error $o(1)$ is small in trace class.  The approximation \eqref{eq:CV-1pdm-GS} was also known in the mean-field regime \cite{NNTwoTerm} (see also \cite{BPS,BLPR} for higher order expansions). 

In this note, our goal is to derive analogous second order approximations for the reduced one- and two-particle density matricies in the setting of positive temperature $T>0$ (of order 1). In this setting, expectation values of observables are determined by the Gibbs state 
 \begin{align} \label{eq:def-Gibbs} \rho_{N,\beta} \equiv \rho_N = \frac1{Z_N} e^{-\beta H_N} \hspace{0.5cm} \text{ for } \hspace{0.5cm} Z_{N,\beta} \equiv Z_N = \tr e^{-\beta H_N}, \end{align}
where $\beta= 1/T$ denotes the inverse temperature. The reduced one- and two-particle density matricies $ \rho_N^{(1)}$ and $ \rho_N^{(2)}$ are defined as the partial traces  
\begin{equation}\label{eq:def-reduced-denstity-matricies}
    \rho_N^{(1)} = \tr_{2,\ldots,N}\,( \rho_N ), \quad \rho_N^{(2)} = \tr_{3,\ldots,N}\,(\rho_N).
\end{equation}

Here is our main result. 

\begin{theorem}\label{thm:main}
Let $0\le V\in L^3 (\bR^3)$ be radial and compactly supported. Consider the Gibbs state in \eqref{eq:def-Gibbs} with $H_N$ given in \eqref{eq:defHN} and $\beta >0$ fixed. For $p\in \Lambda_+^*$, define 
\be
\label{eq:defkp} 
   \epsilon_p  = \sqrt{|p|^4+16\pi\mathfrak{a}|p|^2},   \qquad \mu_p^2  = \frac{|p|^2+8\pi \mathfrak{a} - \epsilon_p}{2\epsilon_p}, \qquad  \theta_p^2 = \frac{|p|^2 +8\pi \mathfrak{a} }{ e^{\beta \epsilon_p}-1}.
\ee 
Then in the limit $N\to \infty$ we have 
        \be \label{eq:1pd}
        \emph{tr}\,\Big| N\rho_N^{(1)} - \Big(N- \sum_{p\in\Lambda_+^*} \big(\mu_p^2 + \theta_p^2\big)\Big)|\ph_0\rangle\langle\ph_0|- \sum_{p\in\Lambda_+^*} \big(\mu_p^2 + \theta_p^2\big)|\ph_p\rangle\langle\ph_p| \Big| \to 0
        \ee
and  
       \begin{align} \label{eq:2pd}
       \emph{\tr} \bigg|  &N\rho_N^{(2)}
         - \Big( N - 4\sum_{p\in\Lambda_+^*}\big(\mu_p^2+\theta_p^2\big)  \Big)|\ph_0\otimes\ph_0\rangle\langle\ph_0\otimes\ph_0|  - 4\sum_{p\in\Lambda_+^*} \big(\mu_p^2+\theta_p^2\big) |\ph_0\otimes\ph_p\rangle\langle\ph_0\otimes\ph_p| \nonumber\\
        &+  4\pi\mathfrak{a}\sum_{p\in\Lambda_+^*}   \bigg(\frac{1}{\epsilon_p} + \frac{2}{e^{\beta\epsilon_p}-1}\bigg) \Big(|\ph_0\otimes\ph_0\rangle\langle\ph_p\otimes\ph_{-p}| +|\ph_{-p}\otimes\ph_p\rangle\langle\ph_0\otimes\ph_{0}|\Big)\bigg| \to 0.
\end{align}
\end{theorem}

Our derivation of Theorem \ref{thm:main} is based on extensions of the arguments in \cite{BBCS3}, allowing to handle eigenvalues and eigenfunctions collectively instead of individually. Besides the information on the one-particle density matrix, we are interested in the two-particle density matrix which is more complicated but more relevant to  correlations in interacting systems. All of this requires a careful revision of the existing techniques.  We will recall some preliminary results in Section \ref{sec:recap}, and then prove \eqref{eq:1pd} and \eqref{eq:2pd} in Sections \ref{sec:1pdm} and \ref{sec:2pdm}, respectively. 

The assumption $V\in L^3(\bR^3)$ can be relaxed and is a consequence of the fact that we use several results from \cite{BBCS3, BBCS4}, for simplicity's sake. In \cite{BBCS3, BBCS4}, the assumption $V\in L^3(\bR^3)$ is used to obtain certain information on the Neumann ground state of $ -\Delta +\frac12 V$ in $ B_\ell(0)\subset \bR^3$, for a small parameter $\ell >0$. Proceeding instead similarly as in \cite{NT, HST} and multiplying $1-f$, where $f$ denotes the solution of the zero energy scattering equation $(-\Delta +\frac12V)f=0 $ with $f(x)\to 1$ as $|x|\to \infty$, by a smooth bump function $\chi \in C^\infty_c(B_\ell(0))$, we expect that our results can be generalized to (radial, compactly supported and non-negative) potentials $V\in L^1(\bR^3)$ without substantial difficulties.  
   
The low temperature  $T\sim 1$ that we are considering is the natural regime to justify Bogoliubov’s prediction on the excitation spectrum. In this case the Gibbs state exhibits the complete Bose-Einstein condensation to the leading order and the temperature effects shows up in the second order terms of the reduced density matrices.  In the higher temperature $T\sim N^{2/3}$, the  Bose-Einstein condensation only holds in a partial sense and the leading order behavior of the one-body density matrix was derived by  Deuchert and Seiringer in \cite{DS} (a similar result in a harmonic trap was proved earlier by Deuchert, Seiringer and Yngvason in \cite{DSY}). It is a very interesting open question to understand the second order correction in this case, which is conceptually related to the understanding of the shift of the critical temperature of the Bose-Einstein condensation under the presence of the interaction.
 
%%%%%%%%%%%%%%%%%%%%%%%%%%%%%%%%%%%%%%%%%%%%%
%%%%%%%%%%%%%%%%%%%%%%%%%%%%%%%%%%%%%%%%%%%%%
%%%%%%%%%%%%%%%%%%%%%%%%%%%%%%%%%%%%%%%%%%%%%
\section{Preliminary Results}\label{sec:recap}
In this section, we introduce basic notation and recall several results from \cite{BBCS3,BBCS4}. 
% that are used in the sequel. %For the corresponding proofs and for further details, we refer the reader to \cite{BBCS3,BBCS4}.
%%%%%%%%%%%%%%%%%%%%%%%%%%%%%%%%%%%%%%%%%%%%%
%%%%%%%%%%%%%%%%%%%%%%%%%%%%%%%%%%%%%%%%%%%%%
%%%%%%%%%%%%%%%%%%%%%%%%%%%%%%%%%%%%%%%%%%%%%
\subsection{Fock Space Formalism} 
Recall that the bosonic Fock space is defined by  
        \[ \cF = \mathbb{C} \oplus \bigoplus_{n =1}^\infty L^2_s (\Lambda^{n}). \]
Particles can be created and annihilated through the usual bosonic creation and annihilation operators on $\cF$. For $g \in L^2 (\Lambda)$, the creation operator $a^* (g)$ and the annihilation operator $a(g)$ are  defined as  
        \[ \begin{split} 
        (a^* (g) \Psi)^{(n)} (x_1, \dots , x_n) &= \frac{1}{\sqrt{n}} \sum_{j=1}^n g (x_j) \Psi^{(n-1)} (x_1, \dots , x_{j-1}, x_{j+1} , \dots , x_n),
        \\
        (a (g) \Psi)^{(n)} (x_1, \dots , x_n) &= \sqrt{n+1} \int_\Lambda  \bar{g} (x) \Psi^{(n+1)} (x,x_1, \dots , x_n) \, dx .  \end{split} \]
They satisfy the canonical commutation relations (CCR), i.e. for $g,h \in L^2 (\Lambda)$, we have
\begin{equation*}%\label{eq:ccr} 
[a (g), a^* (h) ] = \langle g,h \rangle , \quad [ a(g), a(h)] = [a^* (g), a^* (h) ] = 0. \end{equation*}
The vacuum vector in $\cF$ is denoted by $\Omega = \{ 1, 0, \dots \} \in \cF$. In the translation-invariant setting, it is convenient to work in momentum space $\Lambda^* = (2\pi \bZ)^3$. We set
        \begin{equation*}%\label{eq:ap} 
        a^*_p = a^* (\ph_p) \quad \text{and } \quad  a_p = a (\ph_p), \quad \ph_p  (x) = e^{\ii p  x} \in L^2 (\Lambda). 
        \end{equation*} 
        
%For some parts of our analysis, we will switch to position space (where it is easier to use the positivity of the potential $V(x)$). To this end, we introduce operator valued distributions $\check{a}_x, \check{a}_x^*$ defined by 
%\begin{equation*}%\label{eq:axf}
 %a(f) = \int \bar{f} (x) \,  \check{a}_x \, dx , \quad a^* (f) = \int f(x) \, \check{a}_x^* \, dx  \end{equation*}        
In the sequel, we express basic observables in terms of the creation and annihilation operators, which is particularly useful to obtain basic operator and form bounds. First, the number operator $\cN$, defined by  $(\cN\Psi)^{(n)} = n \Psi^{(n)}$ for every $\Psi=(\Psi_n)_{n=0}^\infty \in \cF$, can be written as
        \[ \cN = \sum_{p \in \Lambda^*} a_p^* a_p . \]
Note that by Cauchy-Schwarz and the CCR  
        \begin{equation}\label{eq:abd} 
        \begin{split}
        \| a (f) \Psi \| & =  \bigg(\sum_{p,q \in\Lambda^*} {\widehat f}_q \overline{ \widehat f}_p \langle \Psi, a^*_q a_p\psi \rangle\bigg)^{1/2} \leq   \| f \| \| \cN^{1/2} \Psi \|, \\
        \| a^* (f) \Psi \| & = \| (a^* (f)a(f) + \|f\|^2)^{1/2} \Psi \| \leq \| f \| \| (\cN+1)^{1/2} \Psi \| 
        \end{split}
        \end{equation}
for all $f \in L^2 (\Lambda)$, and hence $a^*(f)$ and $a(f)$ are well-defined on the quadratic form domain of $\cN$. It is also useful to denote by  
$$ \cN_+ = \cN - a_0^*a_0^* = \sum_{p \in \Lambda^*_+} a_p^* a_p$$
the operator that counts the number of excitations (which are orthogonal to the condensate wave function $\ph_0\in L^2(\Lambda)$). The complete BEC \eqref{eq:leading-1pdm} can be written equivalently as 
        \[ \lim_{N\to\infty}  \tr \big| \gamma_N^{(1)} -|\ph_0\rangle\langle\ph_0| \big| = 0\; \Longleftrightarrow \;\lim_{N\to\infty} \langle \ph_0, \gamma_N^{(1)}\ph_0\rangle = 1 \;\Longleftrightarrow\; \lim_{N\to\infty } N^{-1} \langle \psi_N, \cN_+\psi_N\rangle = 0.  \]

Next, we can write the Hamiltonian $H_N$ in \eqref{eq:defHN} in second quantized form
        \begin{equation}\label{eq:HN-fock}
        H_N = \sum_{p \in \Lambda^*} p^2 a_p^* a_p + \frac{1}{2N}\sum_{r,p,q \in \Lambda^*} \wh{V}(r/N) a_{p+r}^* a_q^* a_p a_{q+r}.
        \end{equation}
        with the obvious embedding $L^2_s(\Lambda^N)\hookrightarrow \cF$ and the Fourier transform convention 
$$ \widehat V(p) = \int_{\bR^3} dx\;e^{-\ii p\cdot x}\,V(x), \quad p\in \bR^3.$$

Following Bogoliubov's ideas \cite{Bog}, the mode with $p=0$ is particularly important, because low energy states are expected to exhibit Bose-Einstein condensation into it. For this reason, it is useful to separate the condensate interactions from the interactions among excited particles (that is, particles with non-zero momentum). Our starting point is \cite[Eq. (3.4)]{BBCS3} describing a unitarily equivalent form of $H_N$ which reads
        \begin{equation}\label{eq:cLN}\cL_N = U_N H_N U_N^* = \cK  + \cL_N^{(0)} +\cL_N^{(2)}+\cL_N^{(3)}+\cL_N^{(4)} \end{equation}
 on the  truncated excitation Fock space 
        \[\cF_+^{\leq N} = \bC\oplus\bigoplus_{n=1}^N \bigotimes_{\text{sym}}^n \{\ph_0\}^\bot \hookrightarrow \cF\]
 where
        \begin{equation}\label{eq:cLNj} \begin{split} 
        \cK =\; & \sum_{p \in \Lambda^*_+} p^2 a_p^* a_p, \\
        \cL_N^{(0)} =\;& \frac N2 \wh V (0) -\frac12 \wh V(0) (1-\cN_+/N) - \frac12 \wh V(0) \cN_+^2/N,     \\
        \cL_N^{(2)}=\; & \sum_{p \in \Lambda_+^*} \widehat{V} (p/N)  a_p^* a_p \frac{N-\cN_+}{N} + \frac{1}{2} \sum_{p \in \Lambda^*_+} \widehat{V} (p/N) \left[ b_p^* b_{-p}^* + b_p b_{-p} \right], \\
        \cL_N^{(3)} =\; &\frac{1}{\sqrt{N}} \sum_{p,q \in \Lambda_+^* : p+q \not = 0} \widehat{V} (p/N) \left[ b^*_{p+q} a^*_{-p} a_q  + a_q^* a_{-p} b_{p+q} \right] ,\\
        \cL_N^{(4)} =\; & \frac{1}{2N} \sum_{\substack{p,q \in \Lambda_+^*, r \in \Lambda^*: \\ r \not = -p,-q}} \widehat{V} (r/N) a^*_{p+r} a^*_q a_p a_{q+r}.  
        \end{split} \end{equation}
Here $ U_N: L^2_s(\Lambda^N)\to \cF_+^{\leq N}$ is the unitary map given explicitly by
        \[
        U_N\psi_N  = \bigoplus_{n=0}^{N} (1-|\varphi_0\rangle\langle\varphi_0|)^{\otimes n}
        \frac{a(\varphi_0)^{N-n}}{\sqrt{(N-n)!}} \psi_N
        \]
for $\psi_N\in L^2_s(\Lambda^N)$, as introduced in \cite{LNSS} and 
        \[ b_p^* = a^*_p (1-\cN_+/N)^{1/2}, \quad b_p = (1-\cN_+/N)^{1/2} a_p, \quad p\in\Lambda_+^*,  \]
denote modified creation and annihilation operators introduced in \cite{BS}.  Anticipating that low energy states exhibit complete BEC into the zero momentum mode $\ph_0$ (so that $ a^*_0a_0 \approx N +o(N)$ for such states), one may keep in mind that $b_p^*$ and $ b_p $ act approximately like the usual creation and annihilation operators on such states. As explained in \cite{BS,BBCS1,BBCS2,BBCS3,BBCS4}, the advantage of using $b_p^*$ and $ b_p $ is that they leave the truncated Fock space $\cF_+^{\leq N}$ invariant. A similar version of these  modified creation and annihilation operators was used earlier in the mean-field scaling \cite{Sei,GS}.         

\subsection{Unitary Renormalizations of $\cL_N$} 

In order to understand the spectral property of $\cL_N$ (which is  unitarily equivalent to $H_N$), we use suitable renormalizations introduced in \cite{BBCS3, BBCS4}.  %that are used in the sequel; for recent simplifications and generalizations of the main strategy of \cite{BBCS3}, see in particular \cite{NT, HST}. 
In the first step, we need to extract the leading order contribution  to the ground state energy $E_N$ of $H_N$. 
It is well-known \cite{LY, LSY, LS1, LS2, NRS} that $E_N = 4\pi \mathfrak{a} N + o(N)$ for some error $o(N)$, which is actually bounded uniformly in $N$ \cite{BBCS1, BBCS4,NNRT, BSS1,H,NR}. 
To extract this contribution and to be able to study the fluctuations of $\cL_N$ around it efficiently, it turns out to be useful to renormalize $\cL_N$ through a unitary generalized Bogoliubov transformation $ e^{B_\eta}: \cF_+^{\leq N}\to \cF_+^{\leq N}$ with exponent  
        \be \label{eq:defren1}   B_\eta = \frac12\sum_{p\in \Lambda_+^*} \eta_p \big( b^*_p b^*_{-p} - b_p b_{-p}\big) = -B_\eta^*.  \ee
The kernel $\eta \in \ell^2 (\Lambda^*_+)$ is chosen in order to incorporate relevant short range correlations among the particles. To this end, we will replace the zero-scattering solution in \eqref{eq:scattering-full-space} by a truncated one which is identically equal to $1$ outside a finite ball. To be precise, we consider the ground state solution of the Neumann problem 
        \[ \left[ -\Delta + \frac{1}{2} V \right] f_{\ell} = \lambda_{\ell} f_\ell  \]
on the  ball $|x| \leq N\ell$, normalized so that $f_\ell (x) = 1$ if $|x| = N \ell$. We choose $\ell>0 $ to be sufficiently small, but fixed, and 
we extend $f_\ell$ by one outside the ball $B_{N\ell}(0)$ of radius $N\ell$. By slight abuse of notation, we denote this extension again by $f_\ell$. By scaling, we get 
        \[ 
         \left[ -\Delta + \frac{N^2}{2} V (Nx) \right] f_\ell (Nx) = N^2 \lambda_\ell f_\ell (Nx) \chi_\ell (x), 
        \]
where $\chi_\ell$ denotes the characteristic function of the ball of radius $\ell$. Setting $w_\ell = 1-f_\ell$, which is supported in $B_{N\ell}(0)$, the kernel $\eta \in \ell^2 (\Lambda^*_+)$ is defined as the Fourier transform of
        \[ \Lambda \ni x\mapsto \check{\eta}(x) = - N w_\ell (Nx) \in L^2(\Lambda). \]
One has the basic estimates (see \cite[Lemma 3.1]{BBCS3})
        \be \label{eq:besteta}
        | \check{\eta}(x)| \leq \frac{C\chi_\ell (x)}{|x|+1/N }, \hspace{0.5cm}|\eta_p|\leq \frac{C}{|p|^2},  \hspace{0.5cm} \| \eta\|_2^2 \leq C\ell, 
        \ee
and the coefficients of $\eta=(\eta_p)_{p\in\Lambda_+^*}$ satisfy 
        \[|p|^2 \eta_p + \frac12 \widehat{(Vf_\ell)}(p/N) = N^3 \lambda_\ell \big( \widehat{\chi}_\ell\ast \widehat{f_\ell}(./N) \big)_p.\]
Given $(\eta_p)_{p\in\Lambda_+^*}$ as defined above, we now set
        \[\cG_N = e^{-B_\eta} \cL_N e^{B_\eta},\hspace{0.5cm} \cH_N = \sum_{p\in\Lambda_+^*} a^*_pa_p +\frac{1}{2N} \sum_{\substack{p,q \in \Lambda_+^*, r \in \Lambda^*: \\ r \not = -p,-q}} \widehat{V} (r/N) a^*_{p+r} a^*_q a_p a_{q+r}  = \cK+\cV_N.\]
The next result collects important properties of the renormalized Hamiltonian $\cG_N$ and follows from \cite[Prop. 3.2, Prop. 4.1 \& Lemma 6.1]{BBCS3} and \cite[Proposition 6.1]{BBCS4}. In the following, we write
        \[\lambda_j(A) = \inf_{\substack{ V\subset \cF_+^{\leq N}, \\ \text{dim} V = j } } \sup_{\substack { \psi \in V, \\ \|\psi\|=1 }}  \langle \psi, A \psi\rangle \]
for the $j$-th min-max value of a given semibounded, self-adjoint operator $A$ in $\cF_+^{\leq N}$. Moreover, $ {\bf{1}}_{\Sigma}(A)$ denotes its spectral projection onto a subset $\Sigma\subset \bR$. By standard results, the spectrum of $H_N$ is discrete and its eigenvalues, counted with multiplicity, are equal to the min-max values $\lambda_j(H_N)$, for $j\in\bN$. In the rest of this paper we denote by $ E_N = \lambda_1(H_N) $ the ground state energy of $H_N$. In particular, $ \cP_j(A)\equiv {\bf{1}}_{(-\infty, \lambda_j(A)]}(A)$ corresponds to the projection onto the spectral subspace corresponding to the first $j$ eigenvalues of $A$, counted with multiplicity.

\begin{proposition}[\cite{BBCS3,BBCS4}]\label{prop:GN}
Set $\cG_N'=\cG_N-E_N$, then the following holds true:
\begin{enumerate}[a)]

\item We have that $ E_N = 4\pi  \mathfrak{a}N + O(1)$ and $\cG_N'$ is equal to
\begin{equation}\label{eq:GDelta} \cG_N' = \cH_N + \cE_{\cG_N'} \end{equation}
for an error term $\cE_{\cG_N'}$ which is such that for every $\delta >0$, there exists $C_\delta > 0$ so that 
\begin{equation}\label{eq:Delta-bd} \pm \cE_{\cG_N'} \leq \delta \cH_N + C_\delta (\cN_+ + 1) \, . \end{equation}

\item There exist constants $c, C>0$ so that $\cG_N'$ satisfies the coercivity bound
        \be\label{eq:coerc} \cG_N' \geq  c \,\cK -C.    \ee
        
\item There exist constants $c, C>0$ such that for every $j\in\bN$, we have
        \be \label{eq:mmbndK} c \lambda_j (\cK) - C \leq \lambda_j (\cG_N')\leq C \lambda_j (\cK) + C N^{-1}\lambda_j^{7/2}(\cK). \ee

\item Let $ \cQ_\zeta= \operatorname{ran}\big( {\bf 1}_{[0,  \zeta]} (\cG_N')\big) $ for $\zeta >0$. Then, for every $k\in \bN$, there exists $C_k>0$ so that
        \begin{equation}\label{eq:N+k}
        \begin{split}
        %&\sup_{\substack{ \psi_N = {\bf 1}_{[0,  \zeta]} (H_N')\psi_N  }} \frac{ \langle e^{-B(\eta)} U_N \psi_N, (\cN_+ +1)^k (\cH_N+1) e^{-B(\eta)} U_N \psi_N \rangle }{\| e^{-B(\eta)} U_N\psi_N \|^2 } \\
       &  \sup_{\substack{   \xi_N\in \cQ_\zeta, \|\xi_N\|=1   }}  \langle \xi_N, (\cN_+ +1)^k (\cH_N+1) \xi_N \rangle \leq C_k (1 + \zeta^{k+1}). 
        \end{split}
        \end{equation}
\end{enumerate}
\end{proposition}
\begin{proof}
    Parts $a)$ and $d)$ follow directly from \cite[Propositions 3.2 \& 4.1]{BBCS3}; for part $d)$, see in particular \cite[Eq. (4.4)]{BBCS3}. As observed already in \cite[Eq. (1.21) \& (1.22)]{BBCS3}, the coercivity bound \eqref{eq:coerc} follows from combining the lower bound 
            \[\cL_N \geq 4\pi  \mathfrak{a}N + c \,\cN_+ -C,\]
    which is a direct consequence of \cite[Proposition 6.1]{BBCS4}, with \eqref{eq:Delta-bd} and the fact that 
            \[\sup_{t\in[-1,1]}    e^{-tB_\eta}\cN_+ e^{tB_{\eta}} \leq (1+ C \|\eta\|_2) \cN_+ \]
    for some $C>0$. The latter bound readily follows from a first order Taylor expansion and Gr\"onwall's lemma. Using \eqref{eq:besteta} and recalling that $\ell$ is assumed to small, this implies 
            \[\cG_N' \geq c\, \cN_+ - C\]
    and thus for $\delta, \eps >0$ sufficiently small that 
            \[\begin{split} \cG_N' &\geq \eps\big(  (1-\delta) \,\cH_N - C_\delta(\cN_++1)\big) + (1-\eps)\big( c \,\cN_+ -C\big)\geq  c\, \cH_N -C \geq c \, \cK- C.
            \end{split}\]

To prove part $c)$, note that $ c \lambda_j (\cK) - C \leq \lambda_j (\cG_N')$ is a direct consequence of \eqref{eq:coerc}. The other direction can be proved as in \cite[Lemma 6.1]{BBCS3}: the eigenvalues of $\cK$ are explicit and of the form
        \be\label{eq:evK} \lambda_j(\cK) = \sum_{p \in \Lambda^*_+}  n_p^{(j)} |p|^2 \ee
with coefficients $n_p^{(j)} \in \bN$ which are non-zero for finitely many $p\in\Lambda_+^*$, for every $j\in\bN$. The corresponding normalized eigenvectors can be chosen as 
        \begin{equation}\label{eq:xij} \xi_j = C_j \prod_{p \in \Lambda^*_+} (a^*_p)^{n_p^{(j)}} \Omega, \end{equation}  
for suitable normalization constant $C_j > 0$. If $\xi_j $ is such an eigenvector to eigenvalue $\lambda_j(\cK)$, we have $ n_q^{(j)}=0$ and thus $ a_q\, \xi_j = 0$ whenever $|q| >\lambda_j^{1/2}(\cK)$. As a consequence, we find that
% {\color{blue} Could $r$ be more natural dummy variable here than $u$?}
%         \be \label{eq:VNbnd} \begin{split} 
%         \langle \xi_j, \cV_N \xi_j \rangle &\leq N^{-1} \sum_{p,q,u \in \Lambda_+^* }  |\widehat{V} (u/N)| \| a_{q+u} a_p \xi_j \| \| a_{p+u} a_q \xi \| \\ &\leq CN^{-1}  \sum_{\substack{ p,q,u \in \Lambda_+^* :\\ |p| , |q| , |p+u|, |q+u| \leq \lambda_j^{1/2}(\cK) }}  \| a_{q+u} a_p \xi \| \| a_{p+u} a_q \xi_j \| \\
%         &\leq  C N^{-1} \lambda_j^{3/2}(\cK) \| (\cN_+ + 1) \xi_j \|^2\leq C N^{-1} \lambda_j^{7/2},
%         \end{split} \ee
        \be \label{eq:VNbnd} \begin{split} 
        \langle \xi_j, \cV_N \xi_j \rangle &\leq N^{-1} \sum_{p,q,r \in \Lambda_+^* }  |\widehat{V} (r/N)| \| a_{q+r} a_p \xi_j \| \| a_{p+r} a_q \xi \| \\ &\leq CN^{-1}  \sum_{\substack{ p,q,r \in \Lambda_+^* :\\ |p| , |q| , |p+r|, |q+r| \leq \lambda_j^{1/2}(\cK) }}  \| a_{q+r} a_p \xi \| \| a_{p+r} a_q \xi_j \| \\
        &\leq  C N^{-1} \lambda_j^{3/2}(\cK) \| (\cN_+ + 1) \xi_j \|^2\leq C N^{-1} \lambda_j^{7/2},
        \end{split} \ee
where in the last step we used that $\cN_+^2\leq \cK^2$. Combining this with \eqref{eq:GDelta}, \eqref{eq:Delta-bd} and the min-max principle, we conclude \eqref{eq:mmbndK}.   
\end{proof}

Proposition \ref{prop:GN} provides the leading order contribution to $E_N$, up to an $O(1)$ error, and it proves via \eqref{eq:N+k} a strong quantitative form of BEC into $\ph_0$ for low energy states. To go one step further and determine the order one contributions to the low energy spectrum of $H_N$, one needs to renormalize the cubic contributions (in creation and annihilation operators) to $\cG_N$. Once this is done, one obtains a renormalized excitation Hamiltonian whose order one contributions are fully contained in its quadratic part, and this can be (approximately) diagonalized. We summarize both steps in the next proposition. For a precise statement, we recall a few definitions from \cite{BBCS3}: we define $A_\eta:\cF_+^{\leq N}\to \cF_+^{\leq N}$ by 
        \begin{equation}\label{eq:defA}
        A_\eta = N^{-1/2} \sum_{\substack{ r\in P_H, v\in P_L }} \eta_r \big[ \sinh(\eta_v) b^*_{r+v} b^*_{-r}b^*_{-v} +  \cosh(\eta_v) b^*_{r+v} b^*_{-r}b_{v} - \text{h.c.} \big],           
        \end{equation}
where
        \begin{equation*}
         \hspace{2cm}P_L  = \big\{\, p \in \Lambda_+^* : |p|\leq N^{1/2}\big\} \quad\text{and}\quad 
         P_H =\Lambda_+^* \setminus P_L \nonumber
        \end{equation*}
and $B_\tau:\cF_+^{\leq N}\to \cF_+^{\leq N}$ by
        \be \label{eq:deftau}\begin{split}
        B_\tau &= \frac{1}{2}\sum_{p\in\Lambda_+^*} \tau_p \big( b^*_pb^*_{-p} - b_p b_{-p} \big), \text{ where }  \tau_{p} = - \frac14 \log \big[   \big(1 + 2|p|^{-2}\widehat{(Vf_\ell) } (p/N)\big)\big] - \eta_p. 
        \end{split}
        \ee
It follows from \cite[Lemma 5.1]{BBCS3} that for all $p\in\Lambda_+^*$ we have
        \be \label{eq:taubnd}|\tau_p|\leq C |p|^{-4}.\ee
Conjugation by $ e^{-A_\eta}$ renormalizes the cubic contribution to $\cG_N$ and conjugation by $ e^{-B_\tau}$ approximately diagonalizes $e^{-A_\eta} \mathcal{G}_N e^{A_\eta}$, up to errors of lower order. In the sequel, we abbreviate
        \be\label{eq:defMND}
        \begin{split}
        \cM_N &= e^{-B_\tau}e^{-A_\eta} \mathcal{G}_N e^{A_\eta} e^{B_\tau} , \;\;\;\cD = \sum_{p\in\Lambda_+^*} \epsilon_p a^*_pa_p\;\; \text{ for }\;\;\epsilon_p  = \sqrt{|p|^4+16\pi \mathfrak{a}|p|^2 }.  
        \end{split}\ee

\begin{proposition}\label{prop:MN}
Let $E_N$ be the ground state energy of $\cM_N$ and set $\cM_N'=\cM_N-E_N$. Then:
\begin{enumerate}[a)]

\item We have that 
        \begin{equation}\label{eq:MNid} \cM_N' = \cD  +\cV_N + \cE_{\cM_N'} \end{equation}
for an error term $\cE_{\cM_N}$ which satisfies 
        \begin{equation}\label{eq:EM-bd} 
        \pm \cE_{\cM_N'}\leq CN^{-1/4} \big[ (\cN_+ +1) (\cH_N+1)+(\cN_+ +1)^3\big].
        \end{equation}

\item Let $\zeta >0$ and set $ \cQ_\zeta= \operatorname{ran}\big( {\bf 1}_{[0,  \zeta]} (\cM_N')\big) $. Then, there exists a constant $C>0$ so that
        \begin{equation}\label{eq:cN3}
        \begin{split}
        %&\sup_{\substack{ \psi_N = {\bf 1}_{[0,  \zeta]} (H_N')\psi_N  }} \frac{ \langle e^{-B(\eta)} U_N \psi_N, (\cN_+ +1)^k (\cH_N+1) e^{-B(\eta)} U_N \psi_N \rangle }{\| e^{-B(\eta)} U_N\psi_N \|^2 } \\
        &  \sup_{\substack{   \xi_N\in \cQ_\zeta, \|\xi_N\|=1   }} \big( \langle \xi_N, (\cN_+ +1) (\cH_N+1) \xi_N \rangle+\langle \xi_N, (\cN_+ +1)^3 \xi_N \rangle \big) \leq C (1 + \zeta^{3}). 
        \end{split}
        \end{equation}

\item There exists a constant $ C>0$ such that for every $j\in\bN$, we have
        \be \label{eq:mmbndMN}  \lambda_j(\cD)- \frac{C}{N^{1/4}}\big(1+ \lambda_j^3(\cM_N')\big)  \leq \lambda_j (\cM_N')\leq  \lambda_j(\cD) + \frac{C}{N^{1/4}}\big(1+\lambda_j^3(\cD)\big) + \frac{ C }N\lambda_j^{7/2}(\cD).  \ee
        
\item Fix $ j\in\bN$ so that $   \lambda_j(\cD)< \lambda_{j+1} (\cD ) $. Then, there exists $C>0$, that is independent of $j\in\bN$ and $N\in\bN$, such that for $N$ large enough we have that 
        \[\begin{split}
        \big \| \cP_j(\cM_N') - \cP_j(\cD) \big\|_{\emph{\text{HS}}}^2   &\leq  \frac{C}{N^{1/4}}\frac{\big(\emph{dim}\,\operatorname{ran}\,\cP_j(\cD)\big)}{ (\lambda_{j+1}-\lambda_j)}   \big(1+ \lambda_j  + N^{-1/4} \lambda_j^3   + N^{-1}\lambda_j^{7/2}\big)^3,
        \end{split}\]
where $\lambda_j\equiv \lambda_j(\cD)$. In particular, for every fixed $ \zeta >0$, we have for large enough $N$ that
        \be \label{eq:HSbnd}\big \| \emph{\tbf 1}_{[0, \zeta]}(\cM_N') - \emph{\tbf 1}_{[0, \zeta]}(\cD) \big\|_{\emph{\text{HS}}}^2 \leq \frac{C_\zeta}{N^{1/4}},
        \ee
for some constant $C_\zeta>0$ that depends on $ \zeta>0$, but that is independent of $N$.

\end{enumerate}
\end{proposition}
\begin{proof} The statements in $ a) $ and $b)$ follow from \cite[Corollary 5.5]{BBCS3}. Part $c)$ follows similarly to Prop. \ref{prop:GN} $c)$: the eigenvalues of $\cD$ are explicit and of the form
        \be\label{eq:evD} \lambda_j(\cD) = \sum_{p \in \Lambda^*_+}  n_p^{(j)} \epsilon_p , \ee
with coefficients $n_p^{(j)} \in \bN$ which are non-zero for finitely many $p\in\Lambda_+^*$, for every $j\in\bN$. The corresponding normalized eigenvectors can be chosen in the same way as in \eqref{eq:xij}, i.e. 
        \begin{equation}\label{eq:xij2} \xi_j = C_j \prod_{p \in \Lambda^*_+} (a^*_p)^{n_p^{(j)}} \Omega, \end{equation} 
for a normalization constant $C_j > 0$. Consequently, the bound \eqref{eq:VNbnd} applies to the eigenvectors of $\cD$ as well (using in this case that $\cN_+^2\leq \cD^2$). Together with $a)$ and $b)$, this implies the upper bound on $ \lambda_j(\cM_N')$ in \eqref{eq:mmbndMN}. For the lower bound, we use that $\cV_N\geq 0 $ and part $b)$ to estimate
        \[\begin{split}
        \lambda_j(\cM')= \inf_{\substack{ V\subset \cF_+^{\leq N}, \\ \text{dim} V = j } } \sup_{\substack { \psi \in V, \\ \|\psi\|=1 }}  \langle \psi, \cM_N' \psi\rangle & = \inf_{\substack{ V\subset\, \text{ran}\,\cP_j(\cM_N') , \\ \text{dim} V = j } } \sup_{\substack { \psi \in V, \\ \|\psi\|=1 }}  \langle \psi, \cM_N' \psi\rangle\\
        & \geq  \inf_{\substack{ V\subset \cF_+^{\leq N}, \\ \text{dim} V = j } } \sup_{\substack { \psi \in V, \\ \|\psi\|=1 }}  \langle \psi, \cD \psi\rangle - C N^{-1/4}\big(1+ \lambda_j^3(\cM')\big).
        \end{split} \]
        
Finally, let's explain part $d)$; we follow \cite[Section 7]{GS} (a result similar to \eqref{eq:HSbnd} has been used in \cite[Eq. (6.6))]{BBCS3}; here we recall the key steps of the proof, for completeness). We have
        \[\begin{split}
        \big \| \cP_j(\cM_N') - \cP_j(\cD) \big\|_{ \text{HS}}^2  & = n_j(\cD)+n_j(\cM_N') - 2 \,\tr \cP_j(\cM_N')\cP_j(\cD)
        \end{split}\]
for $ n_j(\cD) =\text{dim}\,\text{ran}\,\cP_j(\cD) $ and $ n_j(\cM_N') =\text{dim}\,\text{ran}\,\cP_j(\cM_N') $. For $j$ fixed, part $c)$ implies that 
        \[ | \lambda_k(\cD) - \lambda_k(\cM_N')| \leq C N^{-1/4 }   \]
for some $C=C_j>0$ and for every $k\leq j+1$ so that $n_j\equiv n_j(\cD) = n_j(\cM_N')$ for $N$ large enough. 

Now, applying once more parts $a), b)$ and  $c)$, we obtain
        \[\begin{split}
        CN^{-1/4} n_j (1+ \lambda_j^3(\cM_N') )& \geq \tr\cP_j(\cM_N')\big(\cD-\lambda_j(\cD) \big)(1-\cP_j(\cD)\big)\\
        &\geq \big(\lambda_{j+1}(\cD)- \lambda_j(\cD)\big) \,\tr\cP_j(\cM_N') (1-\cP_j(\cD)\big)
        \end{split}\]
so that 
        \[ \tr \cP_j(\cM_N')\cP_j(\cD) \geq n_j -\frac{C\,n_j}{N^{1/4}} \frac{ (1+ \lambda_j^3(\cM_N') )}{\big(\lambda_{j+1}(\cD)- \lambda_j(\cD)\big)}.  \]
This proves the first claim. The bound \eqref{eq:HSbnd} is an immediate consequence by choosing some $j\in\bN$ such that $ \lambda_j(\cD)\leq\zeta <\lambda_{j+1}(\cD)$ and applying the first step. 
\end{proof}

\subsection{Conjugation of Basic Observables}

In this section we provide basic results on the conjugation of operators by the unitary maps defined in \eqref{eq:defren1}, \eqref{eq:defA} and \eqref{eq:deftau}. Slightly generalizing \cite[Lemma 2.1 \& Prop. 4.2]{BBCS3}, we first record the following basic lemma (the proof follows with the same arguments as \cite[Lemma 2.1 \& Prop. 4.2]{BBCS3}). 
\begin{lemma} \label{lm:cNkbnds} 
For every $k\in \bR$, there exists $C>0$ such that 
        \[\begin{split} 
        \sup_{t\in [-1,1]} e^{-tB_\eta} (\cN_++1)^k e^{tB_\eta}&\leq C(\cN_++1)^k,\\
        \sup_{t\in [-1,1]} e^{-tA_\eta} (\cN_++1)^k e^{tA_\eta}&\leq C(\cN_++1)^k,\\
        \sup_{t\in [-1,1]} e^{-tB_\tau} (\cN_++1)^k e^{tB_\tau}&\leq C(\cN_++1)^k.  \end{split}\]
\end{lemma}

More importantly in view of the proof of Theorem \ref{thm:main} is the following simple result (related observations have already been used in \cite{RS}).
\begin{lemma} \label{lm:conj1}
Let $ \nu \in \ell^2(\Lambda_+^*)$ be defined by
        \be\label{eq:defnu}\nu_p = - \frac14 \log \big(1+ 16\pi \mathfrak{a} |p|^{-2}\big).   \ee
Then, we have in the sense of forms in $\cF_+^{\leq N}$ that
        \[\begin{split}
        e^{-B_\tau} e^{-A_\eta} e^{-B_\eta} \cN_+  e^{B_\eta} e^{A_\eta} e^{B_\tau} &= e^{-K_{\nu} } \cN_+ e^{K_{\nu} } + \cE_{\cN_+},
        %& = \sum_{q\in\Lambda_+^* } \big(\cosh(\nu_q)a^*_q+\sinh(\nu_q)a_{-q} \big)\big( \cosh(\nu_q)a_q+\sinh(\nu_q)a_{-q}\big)+ \cE,
        \end{split}\]
where $ e^{K_{\nu}}$ is a standard Bogoliubov transformation of the form
        \[K_{\nu} =  \frac12\sum_{p\in\Lambda_+^*} \nu_p \big( a_p^*a^*_{-p} -a_pa_{-p}\big),\]
and where the error $\cE_{\cN_+} $ satisfies $\pm \cE_{\cN_+}\leq CN^{-1/2} (\cN_++1)^{3/2}.$ Similarly, we have for every $q\in\Lambda_+^*$  
        \[\begin{split}
        e^{-B_\tau} e^{-A_\eta} e^{-B_\eta} a^*_q a_q  e^{B_\eta} e^{A_\eta} e^{B_\tau} &= e^{-K_{\nu} } a^*_q a_q e^{K_{\nu} } + \cE_q\\
        %& = \big(\cosh(\nu_q)a^*_q+\sinh(\nu_q)a_{-q} \big)\big( \cosh(\nu_q)a_q+\sinh(\nu_q)a_{-q}\big)+ \cE_q
        \end{split}\]
for some error $\cE_{q} $ that satisfies $ \pm \cE_q\leq CN^{-1/2} (\cN_++1)^{3/2}.$ 
\end{lemma} 

\begin{remark*}
Recall that the action of $ e^{K_{\nu}}$ on creation and annihilation operators is explicit with
        \[\begin{split}
        e^{-K_{\nu}} a_p\, e^{K_{\nu}} & =\cosh(\nu_p) a_p + \sinh(\nu_p)a^*_{-p},\\
        e^{-K_{\nu}} a^*_p \,e^{K_{\nu}} & =\cosh(\nu_p) a^*_p + \sinh(\nu_p)a_{-p}.
        \end{split}\]
\end{remark*}
\vspace{0.2cm}
\begin{proof}
We write
        \[
        \sum_{p\in\Lambda_{+}^{*}}b_{p}^{*}b_{p}=\sum_{p\in\Lambda_{+}^{*}}a_{p}^{*}a_{p}-\frac{1}{N}\sum_{p\in\Lambda_{+}^{*}}a_{p}^{*}\,\cN_{+}a_{p}.
        \]
so that
        \[
        \cN_{+}=\sum_{p\in\Lambda_{+}^{*}}a_{p}^{*}a_{p}=\sum_{p\in\Lambda_{+}^{*}}b_{p}^{*}b_{p}+\cE_{b}
        \]
for some error that satisfies $\pm\cE_{b}\leq (\cN_{+}+1)^{2}/N$. From \cite[Eq. (2.18) and Lemma 2.3]{BBCS3}, we get
        \begin{align*}
        e^{-B_{\eta}}\cN_{+}e^{B_{\eta}} %& =\sum_{p\in\Lambda_{+}^{*}}e^{-B_{\eta}}b_{p}^{*}e^{B_{\eta}}e^{-B_{\eta}}b_{p}e^{B_{\eta}}+e^{-B_{\eta}}\cE_{b}e^{B_{\eta}}\\
         & =\sum_{p\in\Lambda_{+}^{*}}(\cosh(\eta_{p})b_{p}^{*}+\sinh(\eta_{p})b_{p}+d_{p}^{*})(\cosh(\eta_{p})b_{p}+\sinh(\eta_{p})b_{p}^{*}+d_{p})+e^{-B_{\eta}}\cE_{b}e^{B_{\eta}},
        \end{align*}
where, for every $n\in \bZ$, we have
        \begin{align*}
        \|(\cN_{+}+1)^{n/2}d_{p}\xi\| & \leq\frac{C}{N}\left(|\eta_{p}|\|(\cN_{+}+1)^{(n+3)/2}\xi\|+\|b_{p}(\cN_{+}+1)^{(n+2)/2}\xi\|\right)\\
        \|(\cN_{+}+1)^{n/2}d_{p}^{*}\xi\| & \leq\frac{C}{N}\|(\cN_{+}+1)^{(n+3)/2}\xi\|
        \end{align*}
for some constant $C=C_n>0$. Using these estimates, a straightforward application of Cauchy-Schwarz and Lemma \ref{lm:cNkbnds} imply that
        \[
        e^{-B_{\eta}}\cN_{+}e^{B_{\eta}}=\sum_{p\in\Lambda_{+}^{*}}(\cosh(\eta_{p})b_{p}^{*}+\sinh(\eta_{p})b_{p})(\cosh(\eta_{p})b_{p}+\sinh(\eta_{p})b_{p}^{*})+\cE_{B_{\eta}}, 
        \]
for some error $\pm \cE_{B_{\eta}} \leq C (\cN_{+}+1)^{2}/N$.

With regards to the action of $A_\eta$, it is enough to expand 
        \begin{align*}
    	& e^{-A_{\eta}}e^{-B_{\eta}}\cN_+ e^{B_{\eta}}e^{A_{\eta}} = e^{-A_{\eta}}e^{-B_{\eta}}\bigg(\sum_{p\in\Lambda_{+}^{*}}a_{p}^{*}a_{p}\bigg)e^{B_{\eta}}e^{A_{\eta}}\\
    	& =\sum_{p\in\Lambda_{+}^{*}}(\cosh\eta_{p}b_{p}^{*}+\sinh\eta_{p}b_{p})(\cosh\eta_{p}b_{p}+\sinh\eta_{p}b_{p}^{*})\\
    	& \quad+N^{-1/2}\int_0^1\mathrm{d}s\,\sum_{p\in \Lambda_+^*}e^{-sA_{\eta}}\bigg([(\cosh(\eta_{p})b_{p}^{*}+\sinh(\eta_{p})b_{p})(\cosh(\eta_{p})b_{p}+\sinh(\eta_{p})b_{p}^{*}), A_\eta \big]\bigg)e^{sA_{\eta}}\\
     &\hspace{0.5cm}+e^{-A_{\eta}}\cE_{B_{\eta}}e^{A_{\eta}}
        \end{align*}
and to verify with Cauchy-Schwarz the simple bound 
        \begin{align*}
    	& \pm \sum_{p\in \Lambda_+^*}  \big[(\cosh(\eta_{p})b_{p}^{*}+\sinh(\eta_{p})b_{p})(\cosh(\eta_{p})b_{p}+\sinh(\eta_{p})b_{p}^{*}), A_\eta \big] \leq C N^{-1/2} (\cN_{+}+1)^{3/2}.
        \end{align*}
Applying once more Lemma \ref{lm:cNkbnds}, we get 
        \[
        e^{-A_{\eta}}e^{-B_{\eta}}\cN_+ e^{B_{\eta}}e^{A_{\eta}}=\sum_{p\in\Lambda_{+}^{*}}(\cosh(\eta_{p})b_{p}^{*}+\sinh(\eta_{p})b_{p})(\cosh(\eta_{p})b_{p}+\sinh(\eta_{p})b_{p}^{*})+\cE_{A_\eta},
        \]
where $\pm \cE_{A_\eta}\leq CN^{-1/2} (\cN_{+}+1)^{3/2} $.

Finally, arguing as in e.g. \cite[Section 5]{BBCS3} or \cite[Section 5]{HST}, we find that 
        \begin{align*}
        &e^{-B_{\tau}}e^{-A_{\eta}}e^{-B_{\eta}}\cN_{+}e^{B_{\eta}}e^{A_{\eta}}e^{B_{\tau}}\\
        & =\sum_{p\in\Lambda_{+}^{*}}(\cosh(\eta_{p}+\tau_{p})b_{p}^{*}+\sinh(\eta_{p}+\tau_{p})b_{p})(\cosh(\eta_{p}+\tau_{p})b_{p}+\sinh(\eta_{p}+\tau_{p})b_{p}^{*})+\cE_{B_{\tau}}'\\
         & =\sum_{p\in\Lambda_{+}^{*}}(\cosh(\eta_{p}+\tau_{p})a_{p}^{*}+\sinh(\eta_{p}+\tau_{p})a_{p})(\cosh(\eta_{p}+\tau_{p})a_{p}+\sinh(\eta_{p}+\tau_{p})a_{p}^{*})+\cE_{B_{\tau}}
        \end{align*}
for errors $\pm\cE_{B_{\tau}}', \pm\cE_{B_{\tau}}\leq  CN^{-1/2}(\cN_{+}+1)^{3/2} $. Because $|\eta_{p}+\tau_{p}-\nu_{p}|\leq CN^{-1}$ (see e.g. \cite[Section 3]{RS}) uniformly in $p\in\Lambda_+^*$, we conclude that
        \[
        e^{-B_{\tau}}e^{-A_{\eta}}e^{-B_{\eta}}\cN_{+}e^{B_{\eta}}e^{A_{\eta}}e^{B_{\tau}}=e^{-K_{\nu}}\cN_{+}e^{K_{\nu}}+\cE
        \]
for an error $\pm\cE \leq  CN^{-1/2}(\cN_{+}+1)^{3/2} $. The proof w.r.t. the conjugation of $ a^*_q a_q$ is analogous. 
\end{proof}

%%%%%%%%%%%%%%%%%%%%%%%%%%%%%%%%%%%%%%%%%%%%%
%%%%%%%%%%%%%%%%%%%%%%%%%%%%%%%%%%%%%%%%%%%%%
%%%%%%%%%%%%%%%%%%%%%%%%%%%%%%%%%%%%%%%%%%%%%

\section{One-body density matrix}\label{sec:1pdm}
In this section, we derive the second order approximation \eqref{eq:1pd} in Theorem \ref{thm:main} for the one-particle density matrix $\rho_N^{(1)}$ defined in \eqref{eq:def-reduced-denstity-matricies}. We start with a few preliminary remarks. By conservation of the total momentum $ P= \sum_{j=1}^N (-i \nabla_{x_j})$, we can find a joint eigenbasis of $H_N$ and $P$ so that $ \rho_N^{(1)}$ with operator kernel
    \[\rho_N^{(1)}(x;y) = \int_{\Lambda^{N-1}} dX \, \rho_N(x, X)\, \overline{\rho_N}(y,X)\]    
is translation invariant and only depends on $x-y$, for all $x,y\in\Lambda$. Here, we recall that $\rho_N = e^{-\beta H_N}/Z_N $ denotes the Gibbs state at inverse temperature $\beta>0$.

In the following, let us denote by $ N\rho^{(1)}$ the second order approximation of $N\rho_N^{(1)} $ in \eqref{eq:1pd}, namely
        \be\label{eq:lmpd1} N\rho^{(1)} = \Big(N- \sum_{p\in\Lambda_+^*}(\mu_p^2+\theta_p^2)\Big)|\ph_0\rangle\langle\ph_0| + \sum_{p\in\Lambda_+^*}(\mu_p^2+\theta_p^2)|\ph_p\rangle\langle\ph_p|\ee
with $(\mu_p^2)_{p\in\Lambda_+^*}$ and $(\theta_p^2)_{p\in\Lambda_+^*}$ defined in \eqref{eq:defkp}. Below, we will prove that for every $p\in\Lambda_+^*$, it holds true that  
        \be \label{eq:ws1pd} \lim_{N\to\infty} \big| N \tr |\ph_p\rangle\langle \ph_p| \rho_N^{(1)} -  N \tr |\ph_p\rangle\langle \ph_p| \rho^{(1)}\big| =0.  \ee
By translation invariance, i.e. $\tr |\ph_p\rangle\langle \ph_q|\,\rho_N^{(1)} = \tr |\ph_p\rangle\langle \ph_q|\, \rho^{(1)} =0$ for all $p\neq q$, this implies  
        \[ N Q_0 \rho_N^{(1)} Q_0 \stackrel{*}{\rightharpoonup} N Q_0\rho^{(1)}Q_0 = \sum_{p\in\Lambda_+^*}(\mu_p^2+\theta_p^2)|\ph_p\rangle\langle\ph_p|\]
as $N\to \infty$, in the trace class topology. With the same arguments, we prove below that 
        \be \label{eq:tr1pd} \lim_{N\to\infty} \big| N \tr Q_0 \rho_N^{(1)} Q_0-  N \tr Q_0\rho^{(1)}Q_0\big| =0\ee 
so that by positivity of $ Q_0 \rho_N^{(1)} Q_0$ and a standard result \cite[Theorem 2.20]{Simon}, we obtain that
        \[ \lim_{N\to\infty} \tr\big| N Q_0\rho_N^{(1)}Q_0 -N Q_0\rho^{(1)}Q_0 \big| =0.  \]
This implies the desired estimate \eqref{eq:1pd} by noting that
        \[\begin{split}
        N\rho_N^{(1)} = N P_0 \rho_N^{(1)} P_0 + N Q_0 \rho_N^{(1)} Q_0 &= N\big(\tr |\ph_0\rangle\ph_0|\rho_N^{(1)}\big)|\ph_0\rangle\langle\ph_0| +N Q_0 \rho_N^{(1)} Q_0\\
        & = \big(N- N\tr Q_0 \rho_N^{(1)} Q_0\big) |\ph_0\rangle\langle\ph_0| +NQ_0 \rho_N^{(1)} Q_0
        \end{split}\]
        where we have used translation-invariance again in the first equality.  
  
  To conclude \eqref{eq:1pd}, let us therefore focus on the proofs of \eqref{eq:ws1pd} and \eqref{eq:tr1pd}. Since the arguments are quite similar for both statements, we provide a detailed proof for \eqref{eq:tr1pd} only. The analogous steps for the proof of \eqref{eq:ws1pd} are briefly summarized at the end of this section. 

To show \eqref{eq:tr1pd}, we fix from now on $\zeta>0$ (below, we choose $\zeta>0$ large enough, depending on $\beta$, but independently of $N\in \bN$) and we decompose
        \be\label{eq:1pd1}\begin{split}
        \tr  N Q_0\rho_N^{(1)}Q_0 &=  \frac{\tr  \cN_+ \,e^{-\beta H_N}} {\tr e^{-\beta H_N}} \\
        & = \frac{\tr e^{-B_\eta} \cN_+e^{B_\eta} \textbf{1}_{[0,\zeta]}(\cG_N') e^{-\beta \cG_N'} + \tr e^{-B_\eta} \cN_+e^{B_\eta}\textbf{1}_{(\zeta,\infty)}(\cG_N') e^{-\beta \cG_N'}} {\tr \textbf{1}_{[0,\zeta]}(\cG_N') e^{-\beta \cG_N'} + \tr \textbf{1}_{(\zeta,\infty)}(\cG_N') e^{-\beta \cG_N'}}\\
        & = \frac{\tr \cU_{\eta,\tau} \cN_+  \cU_{\eta,\tau}^* \textbf{1}_{[0,\zeta]}(\cM_N') e^{-\beta \cM_N'} + \tr e^{-B_\eta} \cN_+e^{B_\eta}\textbf{1}_{(\zeta,\infty)}(\cG_N') e^{-\beta \cG_N'}} {\tr \textbf{1}_{[0,\zeta]}(\cM_N') e^{-\beta \cM_N'} + \tr \textbf{1}_{(\zeta,\infty)}(\cG_N') e^{-\beta \cG_N'}},
        \end{split}.\ee        
In this expression, we abbreviate $ \cU_{\eta,\tau} = e^{-B_\tau} e^{-A_\eta} e^{-B_{\eta} }$ and we used that $U_N \cN_+ U_N^* = \cN_+$ \cite{LNSS}, where $\cG_N'$ is defined as per equation \eqref{eq:GDelta}.
To estimate the traces over the spectral subspace in which $ \cG_N' > \zeta$, we use Prop. \ref{prop:GN} b) \& c), which imply that 
        \[\begin{split}
          \tr \textbf{1}_{(\zeta,\infty)}(\cG_N') e^{-\beta \cG_N'} = \sum_{j\in \bN: \lambda_j(\cG_N') >\zeta } e^{-\beta \lambda_j(\cG_N')}&\leq  e^{C\beta}\sum_{j\in \bN: \lambda_j(\cK) \geq \max\{\zeta/2, (N\zeta/2)^{2/7}\} } e^{-c\beta \lambda_j(\cK)}\\
          & \leq C \sum_{j\in \bN: \lambda_j(\cK) \geq \zeta/2  } e^{-c \beta  \lambda_j(\cK)}
        \end{split}\]
for sufficiently large $N$ and for suitable positive constants\footnote{Generic constants $C>0$, that are independent of $N$ and $\zeta$, may vary from line to line.} $c, C>0$. Bounding the trace on the r.h.s. in the previous line by the trace over the full excitation Fock space 
        \[\cF_+ =\bC \oplus \bigoplus_{n=1}^\infty \bigotimes_{\text{sym}}^n\big(\{\ph_0\}^\bot\big) \] 
and recalling the spectrum of $ \cK$ from \eqref{eq:evK}, a standard moment bound yields for $\mu>0$ that
        \be\label{eq:mombnd}\begin{split}
            \sum_{j\in \bN: \lambda_j(\cK) \geq \zeta/2  } e^{-c \beta  \lambda_j(\cK)}\leq \prod_{p\in\Lambda_+^*} \sum_{n_p=0}^\infty  e^{-c\mu\zeta/2-c(\beta -\mu) n_p |p|^2} =  e^{-c\mu\, \zeta/2}\prod_{p\in\Lambda_+^* } \frac{1}{1- e^{-c(\beta -\mu)|p|^2 }}.  \end{split}\ee
In particular, choosing $\mu = \beta/2$ and using that 
        \[ -\sum_{p\in \Lambda_+^*}\log \big(1- e^{-c\beta|p|^2/2}\big) = \sum_{p\in\Lambda_+^*}\sum_{n= 1}^\infty n^{-1} e^{-nc\beta|p|^2/2} \leq C,
        \]
we conclude that for suitable constants $c, C>0$ (independent of $N$ and $\zeta$),   
        \be \label{eq:tr>} \tr \textbf{1}_{(\zeta,\infty)}(\cG_N') e^{-\beta \cG_N'} \leq C e^{-c\beta\zeta/4 }= O(e^{-c\beta\zeta/4}). \ee
To bound the error term in the numerator on the r.h.s. in \eqref{eq:1pd1}, we proceed similarly and use the cyclicity of the trace, Lemma \ref{lm:cNkbnds} as well as Prop. \ref{prop:GN} b) to bound 
        \[ \tr e^{-B_\eta} \cN_+e^{B_\eta}\textbf{1}_{(\zeta,\infty)}(\cG_N') e^{-\beta \cG_N'} \leq C\, \tr (\cG_N'+1) \textbf{1}_{(\zeta,\infty)}(\cG_N') e^{-\beta \cG_N'}.  \]
If we assume without loss of generality that $\zeta > 2\beta^{-1} $, so that $x\mapsto (x+1)  e^{-\beta x} $ is monotonically decreasing in $[\zeta/2,\infty)$, we obtain similarly as above the upper bound 
        \[\begin{split}
            \tr (\cG_N'+1) \textbf{1}_{(\zeta,\infty)}(\cG_N') e^{-\beta \cG_N'} &\leq C\, \text{tr}_{\cF_+}\,  \cK \,\textbf{1}_{[\zeta/2,\infty)}(\cK)  e^{-c \beta \cK} \leq C e^{-c\beta \zeta/4}\, \text{tr}_{\cF_+}\, \cK    \,e^{-c \beta \cK/2}.
        \end{split}\]
Combining this with the upper bound
        \[\begin{split}\text{tr}_{\cF_+}\, \cK\,    e^{-c \beta \cK/2} &= \text{tr}_{\cF_+}\,   e^{-c \beta \cK/2}\Big(-\frac2{c\beta}\partial_{\beta}\log \text{tr}_{\cF_+}\,   e^{-c \beta \cK/2}\Big) \\
        &\leq C \Big(\frac2{c\beta}\partial_{\beta}\sum_{p\in \Lambda_+^*}\log \big(1- e^{-c\beta|p|^2/2}\big)\Big) = C \sum_{p\in \Lambda_+^*} \frac{|p|^2 } {  e^{c\beta|p|^2/2}-1 } \leq C, 
        \end{split}\]
we conclude that $\tr e^{-B_\eta} \cN_+e^{B_\eta}\textbf{1}_{(\zeta,\infty)}(\cG_N') e^{-\beta \cG_N'} = O(e^{-c\beta \zeta/4})$ and thus 
        \be\label{eq:1pd2}\begin{split}
        \tr  N Q_0\rho_N^{(1)}Q_0  
        & = \frac{\tr \cU_{\eta,\tau} \cN_+  \cU_{\eta,\tau}^* \textbf{1}_{[0,\zeta]}(\cM_N') e^{-\beta \cM_N'} + O(e^{-c\beta \zeta})} {\tr \textbf{1}_{[0,\zeta]}(\cM_N') e^{-\beta \cM_N'} + O(e^{-c\beta \zeta})}. 
        \end{split}\ee

In the next step, we replace $\cM_N'$ in the traces over the low-energy subspaces in which $\cM_N'\leq \zeta$ by $\cD$. Here, we argue similarly as above and we use additionally the results of Prop. \ref{prop:MN}. Indeed, applying Prop. \ref{prop:MN} c), we obtain first of all that
        \[\tr \textbf{1}_{[0,\zeta]}(\cM_N') e^{-\beta \cM_N'} = \big(1+O(N^{-1/4})\big)\,\tr \textbf{1}_{[0,\zeta]}(\cD) e^{-\beta \cD} \]
for $N$ large enough. Then, using the explicit form \eqref{eq:evD} of the spectrum of $\cD$, we find as in \eqref{eq:mombnd} that 
        \[ \tr \textbf{1}_{[0,\zeta]}(\cD) e^{-\beta \cD} = \tr  e^{-\beta \cD} + O(e^{-\beta \zeta/2}). \]
Similarly, we obtain that 
        \[\begin{split}
           0\leq  \text{tr}_{\cF_+}  e^{-\beta \cD}- \tr e^{-\beta \cD} &=    \sum_{\substack{(n_p)_{p\in\Lambda_+^*}\in(\bN_0)^{\Lambda_+^*}:\\ \sum_{p\in\Lambda_+^*} n_p >N}} e^{-\beta \sum_{p\in\Lambda_+^*}n_p \epsilon_p }  \\
           &\leq e^{-\beta \mu N} \sum_{(n_p)_{p\in\Lambda_+^*}\in(\bN_0)^{\Lambda_+^*} } \prod_{p\in\Lambda_+^*} e^{-\beta (\epsilon_p-\mu) n_p  },
        \end{split}\]
so that, choosing e.g. $0< \mu \leq 4\pi^2 <\inf_{p\in\Lambda_+^*}\epsilon_p$, we conclude 
        \be\label{eq:denexp}\tr \textbf{1}_{[0,\zeta]}(\cM_N') e^{-\beta \cM_N'} = \text{tr}_{\cF_+}   e^{-\beta \cD} + O(e^{-\beta\zeta/2})\ee
for $N$ large enough. We combine analogous arguments with Proposition \ref{prop:MN} c) \& d) and Lemma \ref{lm:conj1} to get  
        \be\label{eq:numexp}\begin{split}
         \tr \cU_{\eta,\tau} \cN_+  \cU_{\eta,\tau}^* \textbf{1}_{[0,\zeta]}(\cM_N') e^{-\beta \cM_N'}& = \tr e^{-K_{\nu}} \cN_+ e^{K_{\nu}} \textbf{1}_{[0,\zeta]}(\cM_N') e^{-\beta \cM_N'} +O\big(N^{-1/4} (1+\zeta^3)\big)  \\
         & = \tr_{\cF_+} e^{-K_{\nu}} \cN_+ e^{K_{\nu}}   e^{-\beta \cD} + O(e^{-\beta \zeta/2})
        \end{split}\ee
for $N$ large enough. Indeed, to obtain the second line, notice that Prop. \ref{prop:MN} d) and induction imply 
        \[ \max_{j\in \bN: \widetilde{\lambda}_j(\cM_N') < \zeta }\| \textbf{1}_{\{\widetilde{\lambda}_j\}}(\cM_N') - \textbf{1}_{\{\widetilde{\lambda}_j\} }(\cD)  \|_{\text{HS}}^2 = \max_{j\in \bN: \widetilde{\lambda}_j(\cD) < \zeta }\| \textbf{1}_{\{\widetilde{\lambda}_j\}}(\cM_N') - \textbf{1}_{\{\widetilde{\lambda}_j\} }(\cD)  \|_{\text{HS}}^2\leq \frac{C_\zeta}{N^{1/4}},  \]
where $ (\widetilde\lambda_j(\cM_N'))$ and $ (\widetilde\lambda_j(\cD))$ denote the eigenvalues of $\cM_N'$ and respectively $\cD$, counted without multiplicity, and where $ C_\zeta$ denotes some constant that depends on $\zeta$, but that is independent of $N$. Using the previous bound, Prop. \ref{prop:MN} c) and the decompositions 
        \[\begin{split}
            \textbf{1}_{[0,\zeta)}(\cM_N') e^{-\beta \cM_N'}& = \sum_{j\in \bN: \widetilde{\lambda}_j(\cM_N') < \zeta} e^{-\beta \widetilde \lambda_j(\cM_N')} \textbf{1}_{\{\widetilde{\lambda}_j\}}(\cM_N'),\\
            \textbf{1}_{[0,\zeta)}(\cD) e^{-\beta \cD}& = \sum_{j\in \bN: \widetilde{\lambda}_j(\cD) < \zeta} e^{-\beta \widetilde \lambda_j(\cD)} \textbf{1}_{\{\widetilde{\lambda}_j\}}(\cD),
        \end{split} \]
we then obtain that 
        \[\big\|  \textbf{1}_{[0,\zeta)}(\cM_N') e^{-\beta \cM_N'}-  \textbf{1}_{[0,\zeta)}(\cD) e^{-\beta \cD}\big\|_{\text{HS}}^2 \leq \frac{C_\zeta}{N^{1/4}}.\]
Combined with Prop. \ref{prop:MN} b), Lemma \ref{lm:conj1} and the fact that $\cN_+^2\leq \cD^2$, this implies  
        \[\begin{split}
       & \big| \tr  e^{-K_{\nu}} \cN_+ e^{K_{\nu}} \textbf{1}_{[0,\zeta]}(\cM_N') e^{-\beta \cM_N'} - \tr e^{-K_{\nu}} \cN_+ e^{K_{\nu}} \textbf{1}_{[0,\zeta]}(\cD) e^{-\beta \cD} \big| \\
       &\leq C\Big( \tr (\cN_++1)^2  \textbf{1}_{[0,\zeta)}(\cM_N') \Big)^{1/2}\big\|\textbf{1}_{[0,\zeta)}(\cM_N') e^{-\beta \cM_N'}-  \textbf{1}_{[0,\zeta)}(\cD) e^{-\beta \cD}\big\|_{\text{HS}} \\
       & \hspace{0.5cm} + \big| \tr e^{-K_{\nu}} \cN_+ e^{K_{\nu}} \textbf{1}_{[\zeta,\infty)}(\cM_N') \textbf{1}_{[0,\zeta)}(\cD) e^{-\beta \cD}\big| \\
       &\leq \frac{C_\zeta}{N^{1/8}} = O(  e^{-\beta \zeta/2})
        \end{split}\]
for $N$ large enough, where in the last step we used that 
        \[ \begin{split}
            &\big| \tr e^{-K_{\nu}} \cN_+ e^{K_{\nu}} \textbf{1}_{[\zeta,\infty)}(\cM_N') \textbf{1}_{[0,\zeta)}(\cD) e^{-\beta \cD}\big|\\
           & = \big| \tr e^{-K_{\nu}} \cN_+ e^{K_{\nu}}\big(  \textbf{1}_{[0,\zeta)}(\cD) -\textbf{1}_{[0,\zeta)}(\cM_N')\big) \textbf{1}_{[0,\zeta)}(\cD) e^{-\beta \cD}\big|\\
           &\leq C_\zeta \big\|\textbf{1}_{[0,\zeta)}(\cM_N')  -  \textbf{1}_{[0,\zeta)}(\cD)  \big\|_{\text{HS}}\leq \frac{C_\zeta}{N^{1/8}}.
        \end{split}\]
Together with this input, similar arguments as in the proof of \eqref{eq:denexp} imply the identity \eqref{eq:numexp}. 

Collecting the above observations and inserting them into \eqref{eq:1pd2}, we have shown that
         \be\label{eq:1pd3}\begin{split}
        \tr  N Q_0\rho_N^{(1)}Q_0  
        & = \frac{\text{tr}_{\cF_+} e^{-K_{\nu}} \cN_+ e^{K_{\nu}}  e^{-\beta \cD} + O(e^{-c\beta \zeta})} {\text{tr}_{\cF_+}   e^{-\beta \cD} + O(e^{-c\beta \zeta})} \\ &=\frac{\text{tr}_{\cF_+} e^{-K_{\nu}} \cN_+ e^{K_{\nu}}  e^{-\beta \cD} } {\text{tr}_{\cF_+}   e^{-\beta \cD} } + O(e^{-c\beta \zeta}). 
        \end{split}\ee
Finally, recalling that an explicit eigenbasis of $\cD$ consists of vectors of the form \eqref{eq:xij2}, which are simultaneous eigenvectors of both $\cD$ and the number of particles operator $\cN_+$, we compute
        \[\begin{split}
          \frac{\text{tr}_{\cF_+} e^{-K_{\nu}} \cN_+ e^{K_{\nu}}  e^{-\beta \cD} } {\text{tr}_{\cF_+}   e^{-\beta \cD} }& =  \sum_{p\in\Lambda_+^*} \sinh^2(\nu_p) + \sum_{p\in\Lambda_+^*}\big(\sinh^2(\nu_p)+ \cosh^2(\nu_p)\big)\frac{\text{tr}_{\cF_+}   a^*_pa_p    e^{-\beta \cD} } {\text{tr}_{\cF_+}   e^{-\beta \cD} }\\
        &=  \sum_{p\in\Lambda_+^*} \sinh^2(\nu_p) - \frac1\beta\sum_{p\in\Lambda_+^*} \cosh(2\nu_p)   \partial_{\omega_p}\Big(\log  \text{tr}_{\cF_+}   e^{-\beta (\cD+\omega_p a^*_pa_p )}\Big)_{|\omega_p=0} \\
        &=  \sum_{p\in\Lambda_+^*} \sinh^2(\nu_p)  + \sum_{p\in\Lambda_+^*} \cosh(2\nu_p) \frac{\epsilon_p}{e^{\beta\epsilon_p}-1 }
        \end{split}\] 
with 
        \[\begin{split}
            \sinh(\nu_p) &=   \frac{|p|^2+8\pi \mathfrak{a} - \sqrt{|p|^4+16\pi \mathfrak{a}|p|^2}}{2\sqrt{|p|^4+16\pi \mathfrak{a}|p|^2}},\;\;\cosh(2\nu_p) = \frac{|p|^{4} +8\pi \mathfrak{a}|p|^{2}}{|p|^{2}\sqrt{|p|^{4}+16\pi \mathfrak{a}|p|^{2}}}.
        \end{split}\]
Comparing this with $\rho^{(1)}$ in \eqref{eq:lmpd1}, the identity \eqref{eq:1pd3} shows that for every $\zeta >2 \beta^{-1}$, we have that
        \[\limsup_{N\to\infty}\big| \tr  N Q_0\rho_N^{(1)}Q_0 - N\tr \rho^{(1)} \big|\leq C e^{-c\beta\zeta} \]
for suitable $c, C>0$ that are independent of $N$ and $\zeta$. Since $\zeta>2\beta^{-1}$ is arbitrary, this proves \eqref{eq:tr1pd}. 

The proof of \eqref{eq:ws1pd} follows along the same lines: in this case we prove for every fixed $p\in\Lambda_+^*$ that
        \begin{align*}
        \tr  N |\ph_p\rangle\langle\ph_p|\rho_N^{(1)} = \frac{\tr a^*_pa_p e^{-\beta H_N}}{\tr e^{-\beta H_N}} &=  \frac{\text{tr}_{\cF_+} e^{-K_{\nu}} a^*_pa_p e^{K_{\nu}}  e^{-\beta \cD} } {\text{tr}_{\cF_+}   e^{-\beta \cD} } + O(e^{-c\beta\zeta}) \\
        &= N\tr   |\ph_p\rangle\langle\ph_p|\rho^{(1)}+O(e^{-c\beta\zeta})
        \end{align*}
and then conclude \eqref{eq:ws1pd} as above. Together with the preliminary remarks at the beginning of this section, this proves the estimate \eqref{eq:1pd} in Theorem \ref{thm:main}. 

%%%%%%%%%%%%%%%%%%%%%%%%%%%%
%%%%%%%%%%%%%%%%%%%%%%%%%%%%
%%%%%%%%%%%%%%%%%%%%%%%%%%%%
 
\section{Two-body density matrix}\label{sec:2pdm}
 
In this section, we derive the second order approximation \eqref{eq:2pd}  in Theorem \ref{thm:main} for the two-particle density matrix $\rho_N^{(2)}$ given in \eqref{eq:def-reduced-denstity-matricies}. 
%We split the proof into two main parts: 
%
%in the following subsection we derive the identity \eqref{eq:1pd} on the one-particle reduced density matrix
%
%and in the final subsection 
In this section, following the proof in Section \ref{sec:1pdm}, we prove the analogous result \eqref{eq:2pd} on the two-particle reduced density matrix.
We use similar arguments as in the derivation of \eqref{eq:1pd}, but deriving \eqref{eq:2pd} is technically slightly more involved, because $\rho_N^{(2)}$ is a two particle operator. Let us outline the overall strategy. Based on the identity
        \[\text{id}_{L^2(\Lambda^2)} = P_0\otimes P_0 + P_0\otimes Q_0+Q_0\otimes P_0+Q_0\otimes Q_0, \]
we split $ \rho_N^{(2)}$ into $\rho_N^{(2)} = \rho^{(2,0)}_{N} + \rho^{(2,2)}_{N} +  \rho^{(2,3)}_{N}+\rho^{(2,4)}_{N} $, where
        \be \begin{split}
        \rho^{(2,0)}_{N}&:= P_0\otimes P_0 \, \rho_N^{(2)}\, P_0\otimes P_0, \\%= \big( \tr |\ph_0\otimes\ph_0\rangle\langle \ph_0\otimes \ph_0| \rho_N^{(2)} \big) |\ph_0\otimes\ph_0\rangle\langle \ph_0\otimes \ph_0| , \\
        \rho^{(2,2)}_{N}&:= 4 P_0\otimes Q_0 \, \rho_N^{(2)}\, P_0\otimes Q_0 + P_0\otimes P_0 \, \rho_N^{(2)}\, Q_0\otimes Q_0 + Q_0\otimes Q_0 \, \rho_N^{(2)}\, P_0\otimes P_0, \\
        \rho^{(2,3)}_{N}&:= 2 P_0\otimes Q_0 \, \rho_N^{(2)}\, Q_0\otimes Q_0 + 2  Q_0\otimes Q_0 \, \rho_N^{(2)}\, P_0\otimes Q_0  ,\\
        \rho^{(2,4)}_{N}&:=  Q_0\otimes Q_0 \, \rho_N^{(2)}\, Q_0\otimes Q_0 .
        \end{split}\ee
Notice that we used the particle exchange symmetry in $L^2_s(\Lambda^2)$ as well as translation invariance, i.e. 
        \[ \tr |\ph_p\otimes \ph_q\rangle\langle \ph_r\otimes \ph_s| \rho_N^{(2)} \equiv 0 \]
whenever $ p+q\neq r+s$, for $p,q,r,s \in \Lambda^* = 2\pi \bZ^3$. In view of the asymptotics \eqref{eq:2pd}, we denote by $\rho^{(2)}$ the second order approximation of $\rho_N^{(2)}$ and we split it accordingly into $\rho^{(2)} =  \rho^{(2,0)} + \rho^{(2,2)}$,
where
        \[ \begin{split}
        \rho^{(2,0)} &= \Big( N - 4\sum_{p\in\Lambda_+^*}\big(\mu_p^2+\theta_p^2\big)  \Big)|\ph_0\otimes\ph_0\rangle\langle\ph_0\otimes\ph_0|, \\
        \rho^{(2,2)} &= \frac4N\!\sum_{p\in\Lambda_+^*} \!\big(\mu_p^2+\theta_p^2\big) |\ph_0\otimes\ph_p\rangle\langle\ph_0\otimes\ph_p| \\
        &\qquad - \frac{4\pi\mathfrak{a}}N\sum_{p\in\Lambda_+^*} \!\!  \bigg(\frac{1}{\epsilon_p} + \frac{2}{e^{\beta\epsilon_p}-1}\bigg) \!\Big(|\ph_0\otimes\ph_0\rangle\langle\ph_p\otimes\ph_{-p}| +\text{h.c.}\Big) , \\
        \end{split}\]
and where we recall that 
        \[\epsilon_p^2= |p|^4 + 16\pi \mathfrak{a}|p|^2,\; \mu_p^2= \frac{|p|^2+8\pi \mathfrak{a} - \epsilon_p}{2\epsilon_p},\;
        \theta_{p}^2=\frac{|p|^{2}+8\pi\mathfrak{a}}{e^{\beta\epsilon_{p}}-1}.  \]
Given the above conventions, the desired estimate \eqref{eq:2pd} follows from the next lemma. 

\begin{lemma} \label{lem:last}Under the same assumptions as in Theorem \ref{thm:main}, the following holds true: 
\begin{enumerate}[a)]
\item  We have that 
        \[ \lim_{N\to\infty} \emph{tr}\big| N \rho^{(2,4)}_{N} \big| = 0.  \]
\item We have that $0\leq  \emph{tr}\, P_0\otimes Q_0 \,N\rho^{(2)} P_0\otimes Q_0 \leq C  $ uniformly in $N$ and that
        \[ \lim_{N\to\infty} \emph{tr}  \big| N \rho_N^{(2,2)} - N\rho^{(2,2)} \big| = 0.   \]
\item As a consequence of $a)$ and $b)$, it follows that
        \[\lim_{N\to\infty} \emph{tr} \big| N \rho^{(2,3)}_{N} \big|  =\lim_{N\to\infty} \emph{tr}  \big| N \rho_N^{(2,0)} - N\rho^{(2,0)} \big| = 0.  \]
In particular, this implies the second order approximation \eqref{eq:2pd} which is equivalent to
        \[\lim_{N\to\infty} \emph{tr}  \big| N \rho_N^{(2)} - N \rho^{(2)}\big|=0. \]
\end{enumerate}
\end{lemma}
\begin{proof}
We start with the proof of part $a)$. Since $  \rho_N^{(2)}\geq 0$, it is enough to show that
    \[ \lim_{N\to\infty} \tr N \rho_N^{(2)} = 0. \]
To this end, we proceed similarly as in Section \ref{sec:1pdm} and we write for fixed $\zeta >0$
        \begin{align} \label{eq:2pd1} 
        &N (N-1)\tr  \rho_N^{(2)}
        =\frac{\tr \cN_+(\cN_+-1) e^{-\beta H_N} }{\tr e^{-\beta H_N}  } \\
        & =\frac{\tr e^{-B_\eta} \cN_+(\cN_+-1) e^{B_\eta} \textbf{1}_{[0,\zeta]}(\cG_N') e^{-\beta \cG_N'} + \tr e^{-B_\eta} \cN_+(\cN_+-1)e^{B_\eta}\textbf{1}_{(\zeta,\infty)}(\cG_N') e^{-\beta \cG_N'}} {\tr \textbf{1}_{[0,\zeta]}(\cM_N') e^{-\beta \cM_N'} + \tr \textbf{1}_{(\zeta,\infty)}(\cG_N') e^{-\beta \cG_N'}}. \nonumber
        \end{align}    
 Now, by \eqref{eq:tr>} and \eqref{eq:denexp}, we have for every $N$ large enough (depending on $\zeta$) that 
        \[\begin{split}       
        \big| \tr \textbf{1}_{[0,\zeta]}(\cM_N') e^{-\beta \cM_N'} + \tr \textbf{1}_{(\zeta,\infty)}(\cG_N') e^{-\beta \cG_N'} - \trf e^{-\beta \cD} \big| \leq C e^{-c\beta \zeta}  
        \end{split}\]
for positive constants $c, C>0$ that are independent of $N$ and $\zeta$. On the other hand, by Lemma \ref{lm:cNkbnds} and Prop. \ref{prop:GN} $d)$, we know that for some $C_\zeta$ that depends on $\zeta$, but that is independent of $N$, we have   
        \[   \tr e^{-B_\eta} \cN_+(\cN_+-1) e^{B_\eta} \textbf{1}_{[0,\zeta]}(\cG_N') e^{-\beta \cG_N'} \leq  C \, \cN_+(\cN_+-1)  \textbf{1}_{[0,\zeta]}(\cG_N') e^{-\beta \cG_N'} \leq C_\zeta  \]
and, by the analysis preceding \eqref{eq:1pd2}, we can upper bound 
        \[\begin{split}  
        &\frac1{N-1}\tr e^{-B_\eta} \cN_+ (\cN_+-1)e^{B_\eta}\textbf{1}_{(\zeta,\infty)}(\cG_N') e^{-\beta \cG_N'} \leq C\; \tr e^{-B_\eta} \cN_+  \textbf{1}_{(\zeta,\infty)}(\cG_N') e^{-\beta \cG_N'} \leq C e^{-c\beta \zeta}.
        \end{split}\]
Combining these observations, we conclude that 
        \[ \limsup_{N\to\infty} \tr N \rho_N^{(2)}\leq   N^{-1} C_\zeta +  C e^{-c\beta \zeta}  \]
for every $N$ large enough. Since $\zeta >0$ is arbitrary, this concludes the proof of part $a)$. 

Let us switch to the proof of part $b)$ and let us split $  \rho_N^{(2,2)} =  \rho_N^{(2,21)}+ \rho_N^{(2,22)}$, where
        \[\begin{split}
        \rho_N^{(2,21)} &:= 4 P_0\otimes Q_0 \, \rho_N^{(2)}\, P_0\otimes Q_0, \\
        \rho_N^{(2,22)} &:=  P_0\otimes P_0 \, \rho_N^{(2)}\, Q_0\otimes Q_0 + Q_0\otimes Q_0 \, \rho_N^{(2)}\, P_0\otimes P_0. 
        \end{split}\]
Analogously, we split $ \rho^{(2,2)} =  \rho^{(2,21)}+ \rho^{(2,22)}$, where
        \[\begin{split}
        \rho^{(2,21)}&=\frac4N\sum_{p\in\Lambda_+^*} \big(\mu_p^2+\theta_p^2\big) |\ph_0\otimes\ph_p\rangle\langle \ph_0\otimes\ph_p|, \\ 
         \rho^{(2,22)}& = - \frac{4\pi\mathfrak{a}}N\sum_{p\in\Lambda_+^*}   \bigg(\frac{1}{\epsilon_p} + \frac{2}{e^{\beta\epsilon_p}-1}\bigg) \!\Big(|\ph_0\otimes\ph_0\rangle\langle\ph_p\otimes\ph_{-p}| +\text{h.c.}\Big),
        \end{split} \]
and we analyze the asymptotics of the contributions $ \rho_N^{(2,21)}$ and $\rho_N^{(2,22)}$ separately. In fact, in the following we prove the stronger statement that 
        \[ \lim_{N\to\infty} \tr  \big| N \rho_N^{(2,21)} - N\rho^{(2,21)} \big| =\lim_{N\to\infty} \tr  \big| N \rho_N^{(2,22)} - N\rho^{(2,22)} \big| = 0,\]
which, combined with the simple upper bound
        \[\tr N\rho^{(2,21)} = \sum_{p\in\Lambda_+^*} \big(\mu_p^2+\theta_p^2\big)\leq C\sum_{p\in\Lambda_+^*} \frac1{|p|^4} \leq C , \]
implies the statements in part $b)$. 

Hence, let us start with the analysis of $\rho_N^{(2,21)} $: by definition, we have that 
        \[\begin{split}
        \rho_N^{(2,21)} &= 4 P_0\otimes Q_0 \, \rho_N^{(2)}\, P_0\otimes Q_0   = 4\sum_{p\in\Lambda_+^*} \big(\tr |\ph_0\otimes\ph_p\rangle\langle \ph_0\otimes \ph_p| \rho_N^{(2)} \big) |\ph_0\otimes\ph_p\rangle\langle \ph_0\otimes \ph_p|  \geq 0 
        \end{split}\]
and similarly that $ \rho^{(2,21)}\geq 0$, so that the convergence 
        \[ \lim_{N\to\infty}\tr \big| N \rho_N^{(2,21)} -  N \rho^{(2,21)}\big| = 0\]
follows if we prove that for all $p\in\Lambda_+^*$, it holds true that
\begin{align}\label{eq:partial-trace-rho2-21}
\lim_{N\to\infty} \tr |\ph_0\otimes\ph_p\rangle\langle \ph_0\otimes \ph_p| \big(N \rho_N^{(2,21)} -N \rho^{(2,21)} \big) =0
\end{align}
and that 
\begin{align}\label{eq:trace-rho2-21}
\lim_{N\to\infty}  \tr  \big(  N \rho_N^{(2,21)}  -N \rho^{(2,21)}   \big) =0.
\end{align}
Like in Section \ref{sec:1pdm}, the arguments for both statements \eqref{eq:partial-trace-rho2-21} and \eqref{eq:trace-rho2-21} are  similar, so let us focus on the proof of \eqref{eq:trace-rho2-21}. Using that 
        \[\begin{split}
        \frac14 \tr N \rho_N^{(2,21)}  = \frac1{N-1} \frac{\tr   a^*_0 a_0\,\cN_+ \, e^{-\beta H_N}}{\tr e^{-\beta H_N} }  = \frac N{N-1} \frac{\tr    \cN_+ \, e^{-\beta H_N}}{\tr e^{-\beta H_N} } - \frac1{N-1} \frac{\tr    \cN_+^2 \, e^{-\beta H_N}}{\tr e^{-\beta H_N} }, 
        \end{split}\]
the proof of the previous part $a)$ and the results of Section \ref{sec:1pdm} imply immediately that 
        \[ \lim_{N\to\infty}\tr N \rho_N^{(2,21)} = 4\lim_{N\to\infty}\frac{\tr    \cN_+ \, e^{-\beta H_N}}{\tr e^{-\beta H_N} }
         = 4\sum_{p\in\Lambda_+^*} \big(\mu_p^2+\theta_p^2\big) = \tr N \rho^{(2,21)}.  \]
Analogously, one shows that for every $p\in\Lambda_+^*$, we have 
        \[\lim_{N\to\infty} \tr |\ph_0\otimes\ph_p\rangle\langle \ph_0\otimes \ph_p| \big(N \rho_N^{(2,21)} -N \rho^{(2,21)} \big) =0\]
and thus $\lim_{N\to\infty}\tr\big|N \rho_N^{(2,21)} - N \rho^{(2,21)} \big|=0$.

To conclude part $b)$, it now only remains to analyze the asymptotics of $\rho_N^{(2,22)}$. Here, we proceed similarly as for $ \rho_N^{(2,21)}$. First of all, it is straightforward to adapt the arguments from Section \ref{sec:1pdm} to show that for every $ p\in \Lambda_+^* $, we have that
        \be \label{eq:2pdaux1} \begin{split}
         & \lim_{N\to\infty} \tr |\ph_0\otimes\ph_0\rangle\langle \ph_p\otimes \ph_{-p}| N \rho_N^{(2,22)} \\
         %&= \lim_{N\to\infty} \frac1{N-1} \frac{\tr a^*_pa^*_{-p} a_0a_0\,e^{-\beta H_N}}{ e^{-\beta H_N}}\\
         & = \lim_{N\to\infty} \frac N{N-1} \frac{\tr e^{-B_\eta} a^*_p\sqrt{1-\cN_{+}/N }a^*_{-p}\sqrt{1-\cN_+/N} e^{B_\eta} \,e^{-\beta \cG'_N}}{\tr  e^{-\beta \cG'_N}}\\
         & = \frac{\trf e^{-K_\nu }a^*_pa^*_{-p}e^{-K_\nu} e^{-\beta\cD } }{\trf e^{-\beta \cD}}= -   4\pi\mathfrak{a}\,\bigg(\frac1{\epsilon_p} + \frac{2 }{e^{\beta\epsilon_p}-1}\bigg).%=\lim_{N\to\infty} \tr |\ph_p\otimes\ph_{-p}\rangle\langle \ph_0\otimes \ph_{0}| N \rho_N^{(2,22)}       
         \end{split} \ee
Now, using that $ \tr | A| =  \tr \sqrt{A^* A} = \tr \sqrt{A A^*}= \tr |A^*|$ for every trace class operator $ A$ (recall that the operators $A^*A$ and $AA^*$ have the same non-zero eigenvalues), we can estimate
        \[\begin{split}
            &\tr \big| N \rho_N^{(2,22)} - N\rho^{(2,22)}\big|\\
            &\leq 2\,\tr |\ph_0\otimes\ph_0\rangle \langle\ph_0\otimes\ph_0| \sqrt{  \sum_{p\in\Lambda_+^*}\bigg| \tr |\ph_0\otimes\ph_0\rangle\langle \ph_p\otimes \ph_{-p}| N \rho_N^{(2,22)} + 4\pi\mathfrak{a}\,\bigg(\frac1{\epsilon_p} + \frac{2 }{e^{\beta\epsilon_p}-1} \bigg) \bigg|^2  }
                    \end{split} \]
                    which vanishes when $N\to\infty$. Indeed, the convergence to zero follows from \eqref{eq:2pdaux1} and dominated convergence, noticing that $ (\epsilon_p^{-1} + 2(e^{\beta \epsilon_p}-1)^{-1} )_{p\in \Lambda_+^*} \in \ell^{2}(\Lambda_+^*)$ and that
        \[\begin{split}
        &\sum_{p\in\Lambda_+^*} \Big| \tr |\ph_0\otimes\ph_0\rangle\langle \ph_p\otimes \ph_{-p}| N \rho_N^{(2,22)} \Big|^2 \\
        &\leq  \sum_{p\in\Lambda_+^*}   \bigg| \frac{\tr e^{-B_\eta} a^*_p\sqrt{1-\cN_{+}/N }a^*_{-p}\sqrt{1-\cN_+/N} e^{B_\eta} \,e^{-\beta \cG'_N}}{\tr  e^{-\beta \cG'_N}}\bigg|^2\\
        & \leq C\frac{\tr \cN_+ e^{-\beta \cG_N'}}{\tr   e^{-\beta \cG_N'}} \sum_{p\in\Lambda_+^*} \frac{ \tr e^{-B_\eta}a^*_p a_pe^{B_\eta} e^{-\beta \cG_N'}}{\tr   e^{-\beta \cG_N'} }  \leq C \bigg(\frac{\tr \cN_+ e^{-\beta \cG_N'}}{\tr   e^{-\beta \cG_N'}} \bigg)^2 \leq C    
        \end{split}\]
for some $C>0$ that is independent of $N$. Here, we used in the second step Cauchy-Schwarz and in the last step the results of Section \ref{sec:1pdm}. Collecting the above observations, this proves part $b)$.  

Finally, let us explain part $c)$. By Cauchy-Schwarz and the previous steps $a)$ and $b)$, we get
        \[ \begin{split} 
            \tr \big | N\rho_N^{(2,3)}\big| &\leq C \,\tr \big | P_0\otimes Q_0 \,N\rho_N^{(2)} Q_0\otimes Q_0  \big|  \leq C \sqrt{ \tr N \rho_N^{(2,21)}   }  \sqrt{ \tr N \rho_N^{(2,4)}   }\\
            &\leq C \Big(1+  \sqrt{ \tr \big(N \rho_N^{(2,21)} - N\rho^{(2,21)}\big)   }\Big)\sqrt{ \tr N \rho_N^{(2,4)}   } \to 0 
        \end{split}\]
in the limit $N\to\infty$. Similarly, combining parts $a)$ and $b)$ with $ \tr \rho_N^{(2)}=1$, $ \tr \rho_N^{(2,22)} = 0$ and
        \[\begin{split}
             &P_0\otimes P_0 N \rho_N^{(2)} P_0\otimes P_0  = \Big( \tr |\ph_0\otimes\ph_0\rangle\langle\ph_0\otimes\ph_0| N\rho_N^{(2)}\Big)|\ph_0\otimes\ph_0\rangle\langle\ph_0\otimes\ph_0|\\
             & = \Big( N - 2\, \tr P_0\otimes Q_0 N\rho_N^{(2)} - \tr Q_0\otimes Q_0\rho_N^{(2)}\Big)|\ph_0\otimes\ph_0\rangle\langle\ph_0\otimes\ph_0| \\
             & = \Big( N -   \tr  N\rho_N^{(2,21)}  - \frac12 \tr  N\rho_N^{(2,22)} -   \tr  N\rho_N^{(2,3)} - \tr  N\rho_N^{(2,4)} \Big) |\ph_0\otimes\ph_0\rangle\langle\ph_0\otimes\ph_0|, 
        \end{split}\]
we conclude the proof of part $c)$ in Lemma \ref{lem:last}. This concludes the estimate \eqref{eq:2pd} in Theorem \ref{thm:main}. 
\end{proof}

%%% Bibliography
%%%%%%%%%%%%%%%%%%%%%%%%%%%%%%%%%%%%%%%%%%%%%
%%%%%%%%%%%%%%%%%%%%%%%%%%%%%%%%%%%%%%%%%%%%%
%%%%%%%%%%%%%%%%%%%%%%%%%%%%%%%%%%%%%%%%%%%%%

\vspace{0.5cm}
\noindent \textbf{Acknowledgements.} This work was partially funded by the Deutsche Forschungsgemeinschaft (DFG, German Research Foundation) via Germany's Excellence Strategy GZ 2047/1, Project Nr. 390685813 (C. Brennecke) and the Beethoven Classics 3 framework, Project Nr. 426365943 (P. T. Nam). J. Lee was   supported by the European Research Council via ERC CoG RAMBAS, Project Nr. 101044249.

\end{document}